\documentclass{article}
\pdfoutput=1

\usepackage[utf8]{inputenc}

\usepackage{amsfonts} 	
\usepackage{amsmath} 	
\usepackage{amssymb}	
\usepackage{bbm}		

\usepackage{csquotes}   
\usepackage{booktabs}	

\usepackage{amsthm}		
\usepackage{comment}    

\usepackage{graphicx}	
\usepackage{subfig}		
\usepackage{float}		
\usepackage{placeins}	

\usepackage{geometry}	

\usepackage{hyperref} 
\usepackage{url}

\newtheorem{theorem}{Theorem}[section]

\newtheorem{corollary}[theorem]{Corollary}
\newtheorem{remark}[theorem]{Remark}

\newtheorem{problem}[theorem]{Problem}

\newcommand{\IN}{\mathbb{N}}
\newcommand{\IR}{\mathbb{R}}
\newcommand{\IZ}{\mathbb{Z}}

\newcommand{\uojdl}{u_{D,0j}^{L}}    
\newcommand{\ujdl}{u_{D,j}^{L}}      
\newcommand{\uojnl}{u_{N,0j}^{L}}    
\newcommand{\ujnl}{u_{N,j}^{L}}      
\newcommand{\vjl}{v_{j}^{L}}        
\newcommand{\zjl}{Z_{j}^{L}}        
\newcommand{\ulmojdl}{u^{l,m}_{D,0j}}    
\newcommand{\ulmjdl}{u^{l,m}_{D,j}}      
\newcommand{\ulmojnl}{u^{l,m}_{N,0j}}    
\newcommand{\ulmjnl}{u^{l,m}_{N,j}}      
\newcommand{\vlmjl}{v^{l,m}_{j}}        
\newcommand{\zlmjl}{Z^{l,m}_{j}}        

\newcommand{\Unormal}{\widehat{\mathbf{n}}}

\newcommand{\NN}{\mathcal{N}}
\newcommand{\radio}{R}
\newcommand{\Zz}{Z}

\newcommand{\NORM}[1]{\left\| #1 \right\|  }
\newcommand{\DUAL}[2]{\langle #1 , #2 \rangle }
\newcommand{\trazaD}{\gamma_D}
\newcommand{\trazaN}{\gamma_N}
\newcommand{\trazas}{\boldsymbol{\gamma}}
\newcommand{\vvector}{\mathbf{v}}

\newcommand{\HHH}{\mathbb{H}}
\newcommand{\Xj}{\operatorname{\mathsf{X}}}
\newcommand{\XXX}{\operatorname{\mathbf{X}}}
\newcommand{\IHDH}{\operatorname{\mathbf{I}_{2\NN \times \NN}}}
\newcommand{\OM}{\operatorname{\mathsf{M}}}
\newcommand{\OOM}{\boldsymbol{\OM}}
\newcommand{\OA}{\operatorname{\mathsf{A}}}
\newcommand{\OOA}{\boldsymbol{\OA}}
\newcommand{\Sconduc}{\boldsymbol\sigma_{\NN \times 4\NN} }

\title{Cell Electropermeabilization Modeling via Multiple Traces Formulation and Time Semi-Implicit Coupling}
\author{Isabel A.~Mart\'inez \'Avila\thanks{Pontificia Universidad Cat\'olica de Chile, Santiago, Chile.}
\and Carlos Jerez-Hanckes\thanks{Facultad de Ingenier\'ia y Ciencias, Universidad Adolfo Ib\'a\~nez, Santiago, Chile.}
\and Irina Pettersson\thanks{Chalmers University of Technology and Gothenburg University, Sweden.}
}

\date{September 1, 2024.}

\begin{document}

\maketitle

\begin{abstract}
We simulate the electrical response of multiple disjoint biological 3D cells undergoing an electropermeabilization process. Instead of solving the boundary value problem in the unbounded volume, we reduce it to a system of boundary integrals equations--the local Multiple Traces Formulation--coupled with nonlinear dynamics on the cell membranes. Though in time the model is highly non-linear and poorly regular, the smooth geometry allows for boundary unknowns to be spatially approximated by spherical harmonics. This leads to spectral convergence rates in space. In time, we use a multistep semi-implicit scheme. To ensure stability, the time step needs to be bounded by the smallest characteristic time of the system. Numerical results are provided to validate our claims and future enhancements are pointed out.

\end{abstract}

Keywords: Electropermeabilization, Electroporation, Multiple Traces Formulation, Boundary integral equations, Semi-implicit scheme, Multistep Methods, Transmembrane potential.    

MSCcodes: 65R20, 92C37, 65M70, 65J15

\section{Introduction}

Electropermeabilization designates the use of short high voltage or electric field pulses to increase the permeability of the cell membrane and its potential to allow the access of non-permeant molecules \cite{KotnikRemsea2019,ROL06}. This technique is used to deliver therapeutic molecules such as drugs and genes into cells to treat cancer, perform genetic engineering, screen drugs, among others applications (cf.~\cite{KimLee2017}).

Theoretically, several models have been proposed to describe the reversible membrane electropermeabilization mechanism without rigorous proof. For instance, during electropermeabilization it is thought that aqueous pores are formed along the cell membrane---electroporation---thereby increasing the permeability of the membrane. Yet, this has not been experimentally observed to occur for voltages used in practice. The pores are either too small to be seen by optical microscopy and too fragile for electron imaging. Only molecular dynamics' simulations have been able to demonstrate pore formation (cf.~\cite[Section 3]{KotnikRemsea2019}, \cite[Section 2.1]{ChoiKhooea2022}). Moreover, the application of external electric pulses triggers other physical and chemical cell mechanisms, many of them not fully understood due to the complex interactions at multiple length scales: from nanometers at the cell membrane to centimeters in tissues \cite{KotnikRemsea2019}. ``{\it Therefore, while the term electroporation is commonly used among biologists, the term electropermeabilization should be preferred in order to prevent any molecular description of the phenomenon}" \cite{ROL06}.

Still, mathematical models and numerical methods can lead to a better understanding of the different underlying phenomena. For instance, Neu and Krassowska \cite{NeuKrassowska1999} consider a pure electroporation process by modeling the nanoscale phenomena involved in the creation and resealing of the cell membrane pores, and apply homogenization theory to derive nonlinear-in-time dynamics. Well-posedness of the Neu-Krassowska model and a new one including anisotropies are derived in \cite{AmmariWidlakea2016}. Alternatively, in \cite{KavianLeguebeea2014} the authors propose a phenomenological model that forgoes the ab initio understanding of the mechanisms involved. A more complete phenomenological model splits the electroporation process into two different stages: conducting and permeable \cite{LeguebeSilveea2014}. This model also takes into account the diffusion and electric transport of non-permeable molecules. In \cite{GuittetPoignardea2017, MistaniGuittetea2019}, the authors discard particle diffusion and transport in \cite{LeguebeSilveea2014} to then apply the Voronoi Interface Method \cite{GuittetLepilliezea2015b} for its numerical approximation. Specifically, they construct a Voronoi mesh of the volume coupled to a ghost fluid method 
to capture discontinuous boundary conditions. Further computational enhancements via parallelization are given in \cite{MistaniGuittetea2019}.
    
Instead of solving the volume boundary value problem, we recast the problem onto cell membranes via the local Multiple Traces Formulation (MTF) \cite{HiptmairJerez2012,CHJ13,HJP14,CHJ15,JPT15}. Originally introduced to solve acoustic wave transmission problems in heterogeneous scatterers, the local MTF considers independent trace unknowns at either side of the subdomains' boundaries to then enforce continuity conditions weakly via Calder\'on identities. In \cite{HJA17, HJH18} the method was successfully applied to model the electrical behavior of peripheral neurons by coupling the Laplace boundary integral operators with Hodgkin-Huxley nonlinear dynamics. The volume Laplace equations in intra- and extracellular media arises when assuming a quasi-static electromagnetic regime and one can show that for 2D and 3D the model is well posed. Numerically, the authors prove stability and convergence of the multistep semi-implicit time discretizations with low- and high (spectral) order spatial boundary unknown representations. Moreover, the numerical method proposed can be extended to model other nonlinear dynamics.

Following \cite{HJA17, HJH18}, we employ the above boundary integral equations to simulate the electric potential response of multiple disjoint cells in three dimensions when subject to electric pulses. Spatially, the boundary unknowns will be approximated by spherical harmonics, thereby allowing for spectral convergence rates in space. The nonlinear dynamics of the cell membrane follow \cite{KavianLeguebeea2014} which are solved by a multistep semi-implicit scheme. As we will see, the dynamic model is in fact highly nonlinear and portrays poor regularity, reason for which convergence rates in time are low when compared to spatial ones.

The rest of the paper is organized as follows. In Section \ref{sec:problem-setup} we formulate the problem and the corresponding non-linear dynamic model, and derive MTF. In Section \ref{sec:numapp}, we present a numerical scheme for spatial and time-domain discretizations, as well as discuss advantages and limitations of the proposed method. Computational results are provided in Section \ref{sec:numres}. Code validation experiments with analytic and overkill solutions confirm our theoretical results and open new avenues of research.


    
\section{Problem Statement and Boundary Integral Formulation}\label{sec:problem-setup}

\subsection{Dirichlet and Neumann traces}
In what follows we will need the notion of Dirichlet and Neumann traces, which we introduce below. Let $\Omega \subset \IR^d$, $d=1,2,3$, be an open non-empty domain not necessarily bounded with a Lipschitz boundary $\Gamma$. For $u\in C^\infty(\overline\Omega)$, Dirichlet and Neumann traces operators are defined as 
$$\trazaD u := u|_{\Gamma},\qquad \trazaN u := \left(\nabla  u \cdot \Unormal\right)|_{\Gamma},$$
where $\Unormal$ is the exterior unit normal and $\nabla$ denotes the gradient. For a Lipschitz $\Gamma$, the Dirichlet trace has a unique extension to a linear and continuous operator $\trazaD: H^1_{loc}(\Omega) \rightarrow H^{1/2}(\Gamma)$, where\footnote{We use the subscript $loc$ for locally integrable spaces, in particular, when $\Omega$ is unbounded (cf.~\cite{McLean2000})}
$$
\|v\|_{H^{\frac{1}{2}}(\Gamma)} :=
\left\{ \inf_{u\in H^1(\Omega)}\|u\|_{H^1(\Omega)}: \ \trazaD u = v\right\}.
$$
The space of bounded linear functionals on $H^{\frac{1}{2}}(\Gamma)$ is denoted by $H^{-\frac{1}{2}}(\Gamma)$. Let $\Delta$ denote the Laplace operator and define 
$$H^1_{loc}(\Delta,\Omega):=\left\{ u \in H^1_{loc}(\Delta,\Omega)\ : \ \Delta u \in L^2_{loc}(\Omega)\right\}.$$
One can also show that
the Neumann trace operator $\trazaN: H^1_{loc}(\Delta,\Omega) \rightarrow H^{-\frac{1}{2}}(\Gamma)$ is  continuous (see  \cite[Section 2.6 to 2.8]{SauterSchwab2010}).
$H^{\frac{1}{2}}(\Gamma)$ and $H^{-\frac{1}{2}}(\Gamma)$ are referred to as Dirichlet and Neumann trace spaces, respectively \cite[Sections 2.4, 2.6 and 2.7]{SauterSchwab2010}.


\subsection{Cell Electropermeabilization Model}
\label{sec:model}

We now present a continuous model used for the electropermeabilization process. Specifically, we assume a quasi-static electromagnetic problem in the intra- and extracellular domains coupled with non-linear dynamics at the cells' membranes. This coupling relies on enforcing adequate transmission conditions for potentials and currents across the cells. By a quasi-static regime, we imply that the frequency of the electric fields is low enough to discard any time delay in electromagnetic wave propagation \cite{PlonseyHeppner1967}.

We consider the electric interaction of $\NN \in\IN$ disjoint spherical cells located at $\mathbf{p_j} \in \IR^3$ with radii $\radio_j\in \IR^+$, $j \in \{1,...,\NN\}$. We define the interior space of the $j$th cell by $\Omega_j:=\{ \mathbf{x} \in \IR^3:\NORM{\mathbf{x}-\mathbf{p_j}}_2<\radio_j \}$, with its membrane being the boundary $\Gamma_j:=\partial \Omega_j=\{ \mathbf{x} \in \IR^3:\NORM{\mathbf{x}-\mathbf{p_j}}_2=\radio_j \}$. The extracellular medium is defined as the complement to the intracellular domain:
$\Omega_0:=\IR^3\setminus\bigcup_{j=1}^\NN\overline{\Omega}_j.$
An illustration for three cells is presented in Figure \ref{Geometria}.

For $j \in \{0,...,N\}$, each cell $\Omega_j$ is assumed to have constant conductivity  $\sigma_j \in \IR^+$. For $T\in \IR^+$, let $\phi_e: [0,T] \times \Omega_0 \rightarrow \IR$ be a given external potential. Let $u_0: [0,T] \times \Omega_0 \rightarrow \IR$ be the electric potential without excitation in the extracellular medium, so that the total external potential is $u_0^{tot}:=u_0+\phi_e$. We denote by $u_j: [0,T] \times \Omega_j \rightarrow \IR$, $j\in \{1,...,\NN\}$, the electric potential inside the $j$th cell, as in Figure \ref{Geometria}.
\begin{figure}[t]
	\begin{center}
		\includegraphics[width=.45\textwidth]{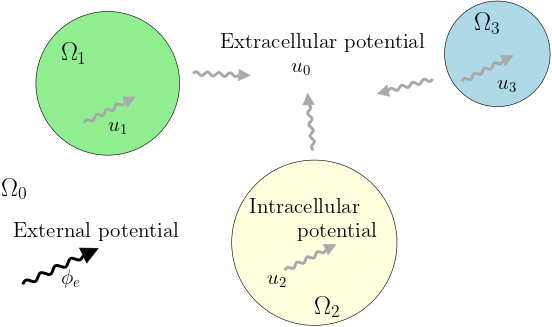}
		\caption{A system of three cells $\NN=3$.}
		\label{Geometria}
	\end{center}
\end{figure}

Across cell membranes $\Gamma_j$ the potential is discontinuous, the difference $v_j:=u_j - u_0$ is called the {\it membrane or transmembrane} potential and the 
flux is assumed to be continuous. Thus, our boundary value problem becomes\footnote{Observe that the Dirichlet and Neumann operators only act in the spatial variable $\mathbf{x}$. For a collection of spheres, we have added super-indices to emphasize where the traces are taken from, i.e., $0j$ for the traces arising from $\Omega_0$ onto $\Gamma_j$, and $j$ for the ones from $\Omega_j$ to $\Gamma_j$.}
\begin{align*}
    \mbox{div} \left(\sigma_j \nabla u_j \right) &=0, & (t,\mathbf{x})\in [0,T] \times\Omega_j, \  j \in \{0,...,\NN\},\\
    -\trazaD^{0j} u_0 +\trazaD^j u_j &=v_j + \trazaD^{0j} \phi_e,  & (t,\mathbf{x})\in  [0,T] \times \Gamma_j, \  j \in \{1,...,\NN\},\\
    \sigma_0\trazaN^{0j}u_0  + \sigma_j \trazaN^{j} u_j &= - \sigma_0 \trazaN^{0j}\phi_e,   & (t,\mathbf{x})\in [0,T] \times\Gamma_j, \   j \in \{1,...,\NN\}.
\end{align*}
Observe that Neumann jumps consider the inherited normal outward direction for each subdomain, so that $\gamma_N^{0j}=-\gamma_N^j$.

For the electropermeabilization process, we adopt the phenomenological model presented in \cite{KavianLeguebeea2014}. Specifically, at each cell $j \in \{1,...,\NN\}$, one has
\begin{align*}
    c_{m,j} \partial_t v_j + I_j^{{ep}}(v_j, \Zz_j) &= - \sigma_j \trazaN^j u_j & \mbox{on } [0,T]\times \Gamma_j,\\
  I_j^{{ep}}(v_j,\Zz_j) &= v_j (S_{L,j} + \Zz_j(t,v_j(t,\mathbf{x})) (S_{ir,j} - S_{L,j}))& \mbox{on } [0,T]\times \Gamma_j,
\end{align*}
with $c_{m,j}$ denoting the membrane capacitance per unit area, and $I_j^{{ep}} $ being the electropermeabilization current. This last quantity  depends on the transmembrane potential $v_j$ and a $C^1$-function $\Zz_j:[0,T] \times \Gamma_j \to [0,1]$ (cf.~\cite[Lemma 7]{KavianLeguebeea2014}). For brevity, and slightly abusing the notations, we write $Z_j(t,\mathbf{x})$ instead of $Z_j(t, v_j(t,\mathbf{x}))$.
{The variable $\Zz_j(t, \mathbf{x})$ ``{\it measures in some way the likelihood that a given infinitesimal portion of the membrane is going to be electropermeabilized}" \cite[p 247]{KavianLeguebeea2014}}. Specifically, $\Zz_j$ enforces the surface membrane conductivity to take values between two parameters: the surface conductivity $S_{ir,j}$ for which the electropermeabilization process becomes irreversible, and the lipid surface conductivity  $S_{L,j}$. 
Indeed, when $\Zz_j=0$, the membrane conductivity equals the lipid conductivity, and there is no electropermeabilization; if $\Zz_j=1$, the membrane conductivity takes the maximal value above which electropermeabilization is irreversible. Following \cite{KavianLeguebeea2014}, {$\Zz_j$ satisfies} the ordinary differential equation\footnote{One can see that functions $Z_j$ will portray poor regularity--at most $C^1$--, a fact later observed numerically.}:
\begin{align*}
    \partial_t \Zz_j(t,\lambda) = \max \left( \frac{\beta_j(\lambda) - \Zz_j(t,\lambda)}{\tau_{ep,j}}, \frac{\beta_j(\lambda) - \Zz_j(t,\lambda)}{\tau_{res,j}}\right).
\end{align*}
Here, $\beta_j\in W^{1,\infty}(\mathbb{R};[0,1]):=\left\{ u \in L^{\infty}(\mathbb{R};[0,1])  \ : \ D^{\alpha} u \in L^{\infty}(\mathbb{R};[0,1]), |\alpha|\leq 1 \right \}$. If $\beta_j(v_j)-\Zz_j(t,v_j)$ is positive, the electric pulse is sufficiently intense to enlarge the electropermeabilized region with a characteristic time $\tau_{ep,j}$. Contrarily, if $\beta_j(v_j)-\Zz_j(t,v_j)$ is negative, the pulse is not strong enough to allow electropermeabilization and the membrane returns to its resting state, with a characteristic resealing time  {$\tau_{res,j}$}. Experimental observations suggest that $\tau_{res,j}\gg\tau_{ep,j}$.
\begin{remark}
    In general \cite{KavianLeguebeea2014}, one can use any function $\beta_j$ such that 
$\beta_j \in W^{1,\infty}(\IR)$, $ v \beta_j'(v) \in L^{\infty}(\IR)$, $\beta_j$ is non decreasing in $(0,\infty)$, $0 \leq \beta_j(v) \leq 1$, $ \lim_{v\to \infty} \beta_j(v) = 1$. In our case, we set $\beta_j$ as 
\begin{align}\label{beta-used}
    \beta_j(v):=\frac{1 + \tanh(k_{ep,j} (|v| - V_{rev,j}) )}{2},
\end{align}
wherein two additional parameters are introduced: the electropermeabilization  switch speed $k_{ep,j}$ between $S_{ir,j}$ and $S_{L,j}$, and $V_{rev,j}$, the transmembrane potential threshold for electropermeabilization to occur. The chosen $\beta_j$ \eqref{beta-used} satisfies the above conditions. This can be checked by recalling the properties of the hyperbolic functions $\tanh: \IR \rightarrow [-1,1]$ and $\text{sech}: \IR \rightarrow [0,1]$
. We will assume that the threshold potential $V_{rev}$ is constant throughout the electropermeabilization process. 
\end{remark}
    In summary, the full electropermeabilization dynamic problem reads:
    \begin{problem}\label{dynamic-problemEP}
        Given  $T\in \IR^+$, an external potential $\phi_e \in C([0,T],H^1_{loc}\left( \Omega_0\right))$,  and the initial conditions $u_j^0 \in H^1\left(\Omega_j\right)$, and $\Zz_{j}^0 \in H^{\frac{1}{2}}(\Gamma_j)$, for $j=1,\ldots,\NN$, we seek $u_j \in C([0,T],H^1\left(\Omega_j\right))$, $v_j \in C([0,T],H^{\frac{1}{2}}(\Gamma_j))$, and $\Zz_{j} \in C([0,T],H^{\frac{1}{2}}(\Gamma_j))$ for $j \in \{1,...,\NN\}$ such that for $t\in[0,T]$, the following holds
        \begin{subequations}
            \begin{align}
           \mathrm{div} \left(\sigma_0 \nabla u_0 \right) &=0  & \mbox{in }\Omega_0,\\
                \mathrm{div} \left(\sigma_j \nabla u_j \right) &=0 & \mbox{in }\Omega_j,\label{CondicionLaplace}\\
                -\gamma _D^{0j} u_0 +\trazaD^j u_j &=v_j + \trazaD^{0j} \phi_e  &\mbox{on }  \Gamma_j, \label{CondicionDirichlet} \\
                \sigma_0\trazaN^{0j} u_0  + \sigma_j \trazaN^{j} u_j &= - \sigma_0 \trazaN^{0j} \phi_e   & \mbox{on } \Gamma_j, \label{CondicionNeumann} \\
                c_{m,j} \partial_t v_j + I_j^{{ep}}( v_j, \Zz_j) &= - \sigma_j \trazaN^j u_j  & \mbox{on } \Gamma_j,\label{nonlinear-condition}\\
                u_j(0,\mathbf{x})=u_j^0, \, Z_j(0,\mathbf{x})&=Z_j^0 & \mbox{in } \Omega_j \label{IC-uj-Z}\\
                u_0(0,\mathbf{x})&=u_0^0 & \mbox{in } \Omega_0, \label{IC-u0}\\
                u_0 & = \mathcal{O}(\|\mathbf{x}\|_2^{-1}) & \mbox{ as }\|\mathbf{x}\|_2\rightarrow \infty , \label{decay-condition}
            \end{align}
        \end{subequations}
        with $I_j^{{ep}}$ defined as:
        \begin{align}
        \label{current}
            I_j^{{ep}}(v_j, \Zz_j) := v_j \left(S_{L,j} + \Zz_j(t,v_j) (S_{ir,j} - S_{L,j}) \right),
        \end{align}
        where the $\Zz_j(t,\lambda)$ satisfy:
        \begin{align}
        \label{ODE}
            \partial_t \Zz_j(t,\lambda) = \max \left( \frac{\beta_j(\lambda) - \Zz_j(t,\lambda)}{\tau_{ep,j}}, \frac{\beta_j(\lambda) - \Zz_j(t,\lambda)}{\tau_{res,j}}\right)
        \end{align}
        with $\beta_j$ given by \eqref{beta-used} and parameters $\tau_{ep,j}, \tau_{res,j}$ described above.
    \end{problem}

As above, we write $Z_j(x,\mathbf{x})= Z_j(t, v_j(t, \mathbf{x}))$. Observe that \eqref{decay-condition} is the standard decay condition for the Laplace problem in three dimensions that guarantees that the problem is well posed.  Finally, the parameters of each cell, $c_{m,j}$, $V_{ep,j}$, $\tau_{ep,j}$ and $\tau_{res,j}$ might differ from cell to cell. In practical applications, these parameters depend on the cell type, e.g., cancer cells possess material properties different from  healthy cells in the same tissue \cite{OnemliJoofea2022}. 
    
\subsection{Boundary integral formulation}
\label{sec:BIOs}
Due to the unboundedness of the domain as well as the constant conductivity values inside intra- and extracellular domains, one can write Problem \ref{dynamic-problemEP} using boundary integral operators, thereby reducing the volume problem to a boundary one as in \cite{HiptmairJerez2012,HJA17,HJH18}. Moreover, this significantly reduces the degrees of freedom required to solve the problem.

\subsubsection{Boundary integral potential and operators}
The free space fundamental solution of the Laplace equation for a source located at $\mathbf{x}'$ satisfying the decay condition \eqref{decay-condition} is (\cite[Section 1.7]{Jackson2013})
\begin{align*}
	g\left(\mathbf{x},\mathbf{x}'\right):= \frac{1}{4\pi \NORM{\mathbf{x}-\mathbf{x}'}_2}, \quad \mathbf{x}\not = \mathbf{x}'. 
\end{align*}
We recall the standard single and double layer operators defined for smooth densities $\psi$ and $\mathbf{x}\in\mathbb{R}^3\setminus \Gamma_j$:
\begin{align*}
	DL_{0j} (\psi)\left(\mathbf{x}\right)&:=	\int_{\Gamma_j}  \psi\left(\mathbf{x}'\right) \nabla g\left(\mathbf{x},\mathbf{x}'\right) \cdot \widehat{\mathbf{n}}_{0j} d\Gamma',&
	SL_{0j} (\psi)\left(\mathbf{x}\right)&:=	\int_{\Gamma_j}  {\psi\left(\mathbf{x}'\right) g\left(\mathbf{x},\mathbf{x}'\right) d\Gamma'},\\ 
	DL_j (\psi)\left(\mathbf{x}\right)&:=	\int_{\Gamma_j}  {\psi\left(\mathbf{x}'\right) \nabla g\left(\mathbf{x},\mathbf{x}'\right) \cdot \widehat{\mathbf{n}}_{j} d\Gamma'}, &
	SL_j (\psi)\left(\mathbf{x}\right)&:=	\int_{\Gamma_j}  {\psi\left(\mathbf{x}' \right) g\left(\mathbf{x},\mathbf{x}'\right) d\Gamma'},
\end{align*}
with the gradient being taken with respect to $\mathbf{r}'$, $\widehat{\mathbf{n}}_j$ being the exterior normal vector of $\Omega_j$. 
Since for each subdomain the normal is oriented outwardly, it holds that $\widehat{\mathbf{n}}_j = -\widehat{\mathbf{n}}_{0j}$. It can be shown that these operators are linear and continuous (cf.~\cite[Section 3.1]{SauterSchwab2010}, \cite[Section 3.1]{HJH18}), in the following Sobolev spaces:
\begin{align*}
  DL_{0j}&: H^{\frac{1}{2}}(\Gamma_j) \rightarrow H^1_{loc} \left(\IR^3 \setminus \cup_{j=1}^{\NN}\Gamma_j\right), & SL_{0j}&: H^{-\frac{1}{2}}(\Gamma_j) \rightarrow H^1_{loc}  \left(\IR^3 \setminus \cup_{j=1}^{\NN}\Gamma_j\right)  , \\
 DL_{j}&: H^{\frac{1}{2}}(\Gamma_j)\rightarrow H^1_{loc} \left(\IR^3 \setminus \cup_{j=1}^{\NN}\Gamma_j\right), & SL_{j}&: H^{-\frac{1}{2}}(\Gamma_j) \rightarrow H^1_{loc}\left(\IR^3 \setminus \cup_{j=1}^{\NN}\Gamma_j\right).
\end{align*}
We will write  $u_j$ in terms of these boundary potentials, and since we aim at rendering Problem \ref{dynamic-problemEP} onto the cells' boundaries, we will take traces of these potentials. This leads to boundary integral operators (BIOs), which are defined by taking the following averages \cite[Section 3.1.2]{SauterSchwab2010}:
\begin{align}
	V_{i,j}^0 &:=  \frac{1}{2} \left( \trazaD^{i} SL_{0j} + \trazaD^{0i} SL_{0j} \right) ,
	& V_{j}&:= \frac{1}{2} \left(  \trazaD^{0j} SL_{j} + \trazaD^{j} SL_{j} \right) ,\nonumber \\
	K_{i,j}^0&:= \frac{1}{2} \left(\trazaD^{i} DL_{0j} + \trazaD^{0i} DL_{0j} \right) ,
	&K_{j}&:= \frac{1}{2} \left(\trazaD^{0j} DL_{j} + \trazaD^{j} DL_{j} \right), \label{BIOS-definition}\\
	K^{*0}_{i,j}&:= \frac{1}{2} \left( - \trazaN ^{i} SL_{0j} + \trazaN ^{0i} SL_{0j}  \right),
	 & K^{*}_{j} &:= \frac{1}{2} \left( -\trazaN ^{0j} SL_{j}  + \trazaN ^{j} SL_{j} \right), \nonumber \\
	W_{i,j}^0 &:= -\frac{1}{2} \left( - \trazaN^{i} DL_{0j}  + \trazaN^{0i} DL_{0j} \right) ,
	& W_{j} &:=- \frac{1}{2} \left( -\trazaN^{0j} DL_{j} + \trazaN^{j} DL_{j} \right). \nonumber
\end{align}
One can show that these operators are linear and continuous \cite[Theorem 3.1.16]{SauterSchwab2010} in the following Sobolev spaces: 
\begin{align*}
	V_{{i},j}^0 &: H^{-\frac{1}{2}}(\Gamma_j) \rightarrow H^{\frac{1}{2}}(\Gamma_i),
	&V_{j}&: H^{-\frac{1}{2}}(\Gamma_j) \rightarrow H^{\frac{1}{2}}(\Gamma_j),\\
	W_{{i},j}^0&: H^{\frac{1}{2}}(\Gamma_j) \rightarrow H^{-\frac{1}{2}}(\Gamma_i),
	&W_{j}&: H^{\frac{1}{2}}(\Gamma_j) \rightarrow H^{-\frac{1}{2}}(\Gamma_j) ,\\
	K_{{i},j}^0&: H^{\frac{1}{2}}(\Gamma_j) \rightarrow H^{\frac{1}{2}}(\Gamma_i),
	&K_{j}&: H^{\frac{1}{2}}(\Gamma_j) \rightarrow H^{\frac{1}{2}}(\Gamma_j) ,\\
	K^{*0}_{{i},j}&: H^{-\frac{1}{2}}(\Gamma_j) \rightarrow H^{-\frac{1}{2}}(\Gamma_i),
	&K^*_{j}&: H^{-\frac{1}{2}}(\Gamma_j) \rightarrow H^{-\frac{1}{2}}(\Gamma_j).
\end{align*}
For smooth domains, the jump relations for the potentials across a closed boundary \cite[Theorem 3.3.1]{SauterSchwab2010} yield 
\begin{align*}
	V_{{i},j}^0 &=   \trazaD^{0{i}} SL_{0j},
	& V_{j}&=  \trazaD^{j} SL_{j},\\
	W_{{i},j}^0 &=-  \trazaN^{0{i}} DL_{0j},
	& W_{j} &=- \trazaN^{j} DL_{j}, \\
	  K_{{i},j}^0&= \trazaD^{0{i}} DL_{0j}\mbox{ with } {i} \not=j,
	 & K^{*0}_{{i},j} &= \trazaN ^{0{i}} SL_{0j}\mbox{ with } {i}  \not=j,\\
	K_{j,j}^0(\psi) &= \frac{1}{2}\psi +\trazaD^{0j} {DL_{0j}(\psi)} ,
	&K_{j}(\psi) &= \frac{1}{2} \psi +\trazaD^{j} {DL_{j}(\psi)} ,\\
	 K^{*0}_{j,j}(\psi) &= -\frac{1}{2} \psi + \trazaN^{0j} {SL_{0j}(\psi)},
	&K^*_{j}(\psi) &= -\frac{1}{2} \psi + \trazaN^j {SL_{j}(\psi)}.
\end{align*}
The next theorem allows us to reconstruct the electric potentials $u_j$ and $u_0$ from boundary layer operators. 
\begin{theorem}(\cite[Section 3]{SauterSchwab2010})\label{FormulaRepresentacion}
The integral representation formulas for $u_j \in H^1(\Omega_j)$, $u_0 \in H^1_{loc}(\Omega_0)$ read
    \begin{subequations}
    	\begin{align}
    	\label{u0representacion}
    	u_0 &= - \sum_{{i}=1}^\NN DL_{0{i}} \left(\trazaD^{0{i}} u_0 \right) + \sum_{{i}=1}^\NN SL_{0{i}} \left(\trazaN^{0{i}}  u_0 \right) ,\\
    	\label{ujrepresentacion}
    	u_j&=	-DL_j \left(\trazaD^{j}  u_j \right) +SL_j \left(\trazaN^{j}  u_j  \right), \quad \forall j\in \{1, ..., \NN \}.
    \end{align}
    where $u_j$ are zero-valued on the complement of $\Omega_j$, for $j=0,\ldots,\NN$.
    \end{subequations}
\end{theorem}

\subsubsection{Multiple traces formulation for Problem \ref{dynamic-problemEP}}
We now write the MTF of Problem \ref{dynamic-problemEP} by taking traces of the integral representation formula (cf.~\cite{HiptmairJerez2012} and later references). 

For $j\in \{1,...,\NN \}$, we introduce the  Cartesian product of Hilbert  spaces $ {\boldsymbol{H}_j} := H^{\frac{1}{2}}(\Gamma_j) \times H^{-\frac{1}{2}}(\Gamma_j)$, with graph norm $\|\cdot \|_{\boldsymbol{H}_j}= \|\cdot \|_{H^{\frac{1}{2}}(\Gamma_j)} + \| \cdot \|_{H^{-\frac{1}{2}}(\Gamma_j)}$. Let be $\boldsymbol{\phi}$, $\boldsymbol{\xi}$ $\in {\boldsymbol{H}_j}$, such that $\boldsymbol{\phi} = (\phi_D, \phi_N )$ and $\boldsymbol{\xi} = (\xi_D, \xi_N )$. We introduce the cross-product over $\Gamma_j$ \cite[Section 2.2.1]{HiptmairJerez2012} by $\DUAL{\boldsymbol{\phi}}{\boldsymbol{\xi}}_{\times,j} :=     \DUAL{\phi_D}{ \xi_N }_{j} + \DUAL {\xi_D }{\phi_N}_{j},$ where for brevity we denote $\DUAL {\xi_D }{\phi_N}_{j}$ := $\DUAL {\xi_D }{\phi_N}_{H^{\frac{1}{2}}(\Gamma_j) \times H^{-\frac{1}{2}}(\Gamma_j)}$.

We define also the flip-sign operator ${\Xj}_j:{\boldsymbol{H}_j}\rightarrow {\boldsymbol{H}_j}$, ${\trazas}^{0j}:H^1_{loc}(\Delta,\Omega_0)\rightarrow {\boldsymbol{H}_j}$ and ${\trazas}^{j}:H^1(\Delta,\Omega_j)\rightarrow {\boldsymbol{H}_j}$ as:
\begin{equation}
	\label{Notacion1}
	{\Xj}_j:=	\begin{bmatrix}
    	I & 0 \\[2mm]
        \displaystyle
    	0 & -\frac{\sigma_0}{\sigma_j} I
	\end{bmatrix}, \quad
	{\trazas}^{0j} := 	\begin{pmatrix}
    	\trazaD^{0j} \\[2mm]
    	\trazaN^{0j}
	\end{pmatrix}
	\; \mbox{ and } \;
	{\trazas}^{j} := 	\begin{pmatrix}
    	\trazaD^{j} \\[2mm]
    	\trazaN^{j}
	\end{pmatrix},\quad j \in \{1,..., \NN\},
\end{equation}
with $I$ being the identity operator understood in the corresponding functional space. Then, we write Dirichlet and Neumamn boundary conditions, \eqref{CondicionDirichlet} and \eqref{CondicionNeumann}, respectively, succinctly as
\begin{subequations}
	\begin{align}
		\label{CondicionBorde}
	-{\Xj}_j{\trazas}^{0j} u_0+{\trazas}^j u_j &= {\Xj}_j (v_j , \  0)^t + {\Xj}_j {\trazas}^{0j} \phi_e , \\
		\label{CondicionBorde2}
	{\trazas}^{0j} u_0 -{\Xj}_j^{-1}{\trazas}^j u_j &=- (v_j , \  0)^t - {\trazas}^{0j} \phi_e,	
	\end{align}
\end{subequations}
with superscript $t$ denoting transposition, and where both equations are equivalent. Taking Dirichlet and Neumann traces of both  \eqref{u0representacion} and \eqref{ujrepresentacion}, and rewriting in terms of BIOs, we obtain
\begin{align*}
	\trazaD^{0j} u_0 &=	-\left(-\frac{1}{2}I\left(\trazaD^{0j} u_0 \right)+\sum_{i=1}^n K_{j,i}^0 \left(\trazaD^{0i} u_0 \right)\right) +\sum_{i=1}^n V_{j,i}^0 \left(\trazaN^{0i} u_0 \right),\\
	\trazaN^{0j}  u_0 &=	\sum_{i=1}^n W_{j,i}^0 \left(\trazaD^{0i} u_0 \right) +\left( \frac{1}{2}I\left(\trazaN^{0j} u_0 \right)+\sum_{i=1}^n K_{j,i}^{*0} \left(\trazaN^{0i} u_0 \right)\right),\\
	\trazaD^{j} u_j &=	-\left(-\frac{1}{2}I\left(\trazaD^{j} u_j \right) +K_j\left(\trazaD^{j} u_j \right) \right) 
	+V_j \left(\trazaN^{j} u_j \right),\\
	\trazaN^{j} u_j &=	W_j \left(\trazaD^{j} u_j \right) +\left( \frac{1}{2}I\left(\trazaN^{j} u_j \right) + K^*_j\left(\trazaN^{j} u_j \right) \right) .
\end{align*}
After some algebra, one can write
\begin{equation*}
{\trazas}^{0j} u_0=	2\sum_{i=1}^\NN \OA_{j,i}^0 {\trazas}^{0i}u_0, \quad {\trazas}^j u_j=	2\OA_j {\trazas}^j u_j, \quad  j \in \{1,...,\NN \}, 
\end{equation*}
with $\displaystyle \OA_{j,i}^0:=	\begin{bmatrix}
    	-K_{j,i}^0 & V_{j,i}^0 \\
    	W_{j,i}^0 & K_{j,i}^{*0} \\
	\end{bmatrix} \mbox{ and } \OA_j:=	\begin{bmatrix}
    	-K_j & V_j \\
    	W_j & K_j^* \\
	\end{bmatrix}.$ By replacing ${\trazas}^{0j} u_0$, ${\trazas}^j u_j$ into \eqref{CondicionBorde2} and \eqref{CondicionBorde}, we obtain
\begin{align*}
    2\sum_{i=1}^n \OA_{j,i}^0 {\trazas}^{0{i}} u_0 
    -{\Xj}_j^{-1}{\trazas}^j u_j 
    &=- (v_j , \  0)^t - {\trazas}^{0j} \phi_e  ,\\
	-{\Xj}_j{\trazas}^{0j} u_0 +2\OA_j {\trazas}^j  u_j &= {\Xj}_j (v_j , \  0)^t + {\Xj}_j {\trazas}^{0j} \phi_e  \quad \mbox{ on } \Gamma_j.
\end{align*}

We define the Cartesian product space of multiple traces $\HHH:=\Pi_{j=1}^\NN \boldsymbol{H}_j$ and  $\HHH^{(2)} := {\HHH} \times \HHH = \Pi_{j=1}^{\NN} \boldsymbol{H}_j \times \Pi_{j=1}^{\NN} \boldsymbol{H}_j$; the multiple trace spaces reordering  $\HHH_D:=\Pi_{j=1}^\NN H^{\frac{1}{2}}(\Gamma_j)$, $\HHH_N:=\Pi_{j=1}^\NN H^{-\frac{1}{2}}(\Gamma_j)$; and, the cross-product $$\displaystyle \langle \boldsymbol{\phi}, \boldsymbol{\xi } \rangle_{\times} = \sum_{j=1}^\NN \langle \boldsymbol{\phi}^{0j}, \boldsymbol{\xi}^{0j} \rangle_{\times,j} + \sum_{j=1}^\NN \langle \boldsymbol{\phi}^{j}, \boldsymbol{\xi}^{j} \rangle_{\times,j},$$ with $\boldsymbol{\phi} = (\boldsymbol{\phi}^{01},..., \boldsymbol{\phi}^{0\NN}, \boldsymbol{\phi}^{1}, ..., \boldsymbol{\phi}^{\NN})$ and $\boldsymbol{\xi} = (\boldsymbol{\xi}^{01},..., \boldsymbol{\xi}^{0\NN}, \boldsymbol{\xi}^{1}, ..., \boldsymbol{\xi}^{\NN})$.

The local MTF operator \cite[Section 3.2.3]{HiptmairJerez2012} $\OOM_\NN: \HHH^{(2)} \rightarrow \HHH^{(2)}$ for the problem presented in Section \ref{sec:problem-setup} takes the form
\begin{align}\label{operator-M}
	\OOM_{\NN}:=	\begin{bmatrix}
        2\OOA_{0,\NN} & -{\XXX}_{\NN}^{-1} \\
        -{\XXX}_{\NN} & 2\OOA_{1,\NN}
	\end{bmatrix}, \mbox{ with } \OOA_{0,\NN}:=\begin{bmatrix}
    	\OA_{1,1}^0 & \OA_{1,2}^0 & ... & \OA_{1,\NN}^0\\
    	\OA_{2,1}^0 & \OA_{2,2}^0 & ... & \OA_{2,\NN}^0\\
    	\vdots &  &\ddots & \vdots\\
    	\OA_{\NN,1}^0 & \OA_{\NN,2}^0 & ... & \OA_{\NN,\NN}^0
	\end{bmatrix},
\end{align}
and diagonal operators
$\OOA_{1,\NN}:=\mathrm{diag}\left(\OA_{1},\ldots, \OA_{\NN}\right)$ and $\XXX_\NN:=\mathrm{diag}\left(
    	{\Xj}_{1},\ldots {\Xj}_{\NN}\right)$.
With the MTF operator, the interface conditions \eqref{CondicionLaplace}, \eqref{CondicionDirichlet} and \eqref{CondicionNeumann} (Problem \ref{dynamic-problemEP}) can be written as:
\begin{align}\label{MTF1}
	{\OOM}_\NN \begin{pmatrix} \trazas^0_u \\	 \trazas_u \end{pmatrix} =  \begin{pmatrix} - \IHDH \vvector \\ \XXX_{\NN} \IHDH \vvector	\end{pmatrix} + \begin{pmatrix} -\trazas^0_{\phi_e} \\ \XXX_{\NN} \trazas^0_{\phi_e}	\end{pmatrix},
\end{align}
where we use the notation:
\begin{align*}
	\trazas^0_u &:=\begin{pmatrix} \trazas^{01} u_0, \trazas^{02} u_0, \hdots, \trazas^{0\NN} u_0 \end{pmatrix}^t , &\trazas_u &:=\begin{pmatrix} \trazas^1 u_1 , \trazas^2 u_2, ..., \trazas^\NN u_\NN \end{pmatrix}^t,\\
	\trazas^0_{\phi_e} &: =\begin{pmatrix} \trazas^{01} \phi_e, \trazas^{02} \phi_e, \hdots, \trazas^{0\NN} \phi_e \end{pmatrix}^t, & \vvector &:=\begin{pmatrix} v_1, v_2, v_3, \hdots , v_\NN \end{pmatrix}^t,
\end{align*}
and the operator $\IHDH: \HHH_D \rightarrow \HHH$ defined as
\begin{align*}
    \IHDH := \begin{pmatrix}
    I & 0 &...& 0 \\
    0 & 0 & ... & 0 \\
    0 & I & ... & 0 \\
    0 & 0 & ... & 0 \\
    \vdots & \vdots & & \vdots \\
    0 & 0 & ... & I \\
    0 & 0 & ... & 0 
    \end{pmatrix}.
\end{align*}
Notice that the identity operators act on the corresponding Dirichlet traces.

The following result is a consequence of \cite[Proposition 3.9, Proposition 3.10]{HJH18} along with the Fredholm alternative.

\begin{theorem}[Existence, uniqueness and stability] \label{ex-un-st}
    The operator $\OOM_\NN$ is a linear, injective and coercive operator in ${\HHH}^{(2)}$. For all $\boldsymbol{\xi} \in {\HHH^{(2)}}$, there exists a unique weak solution $\boldsymbol\lambda \in {\HHH^{(2)}}$ of 
        \begin{align*}
        	( \OOM_\NN \boldsymbol\lambda , \boldsymbol{\phi} )_{\times} = ( \boldsymbol{\xi} , \boldsymbol{\phi} )_{\times}, \quad \forall \boldsymbol{\phi} \in {\HHH^{(2)}},
        \end{align*}
    that satisfies the stability estimate $\displaystyle
\|\boldsymbol\lambda\|_{{\HHH^{(2)}}}\leq  c \|\boldsymbol\xi\|_{{\HHH^{(2)}}}$, for a constant $c>0$.
\end{theorem}

\subsubsection{Boundary integral formulation of Problem \ref{dynamic-problemEP}}

Until this point, we have not introduced the membrane dynamics of Problem \ref{dynamic-problemEP}. In the following, we will use the theory presented in \cite{HJA17, HJH18} to combine the MTF with the nonlinear dynamics. Indeed, thanks to Theorem \ref{ex-un-st} we can take the inverse of the MTF operator, and \eqref{MTF1} becomes
\begin{align*}
     \begin{pmatrix} \trazas^0_u \\[2mm] \trazas_u \end{pmatrix}  = \OOM_\NN^{-1}  \begin{pmatrix} - \IHDH \vvector \\[2mm] \XXX_\NN \IHDH \vvector	\end{pmatrix} + \OOM_\NN^{-1} \begin{pmatrix} -\trazas^0_{\phi_e} \\[2mm] \XXX_\NN \trazas^0_{\phi_e} \end{pmatrix}.
\end{align*}

The even components of the vector $\gamma_u$ (the interior Neumann traces), related to the nonlinear dynamics of the problem by \eqref{nonlinear-condition}, can be retrieved as follows:
\begin{align*}
   \begin{pmatrix}
   \sigma_1 \trazaN^1(u_1) \\
   \sigma_2 \trazaN^2(u_2) \\
   \vdots \\
   \sigma_\NN \trazaN^\NN(u_\NN) \\
   \end{pmatrix}
   =  \boldsymbol\sigma_{\NN \times 4\NN} \OOM_\NN^{-1}\left( \begin{pmatrix} - \IHDH \vvector \\ \XXX_\NN \IHDH \vvector	\end{pmatrix} + \begin{pmatrix} -\trazas^0_{\phi_e} \\ \XXX_\NN \trazas^0_{\phi_e} \end{pmatrix} \right) ,
\end{align*}
where the dimensions of $\boldsymbol\sigma_{\NN \times 4\NN}$ are $\NN \times 4\NN$, the first half containing only zeros:
\begin{align*}
\boldsymbol\sigma_{\NN \times 4\NN}:=
\begin{pmatrix}
   0 & ... & 0 & \sigma_1I & 0 & 0 & ... & 0 \\
   0 & ... & 0 & 0 & 0 & \sigma_2I & ... & 0 \\
   \vdots & & \vdots & \vdots & \vdots & \vdots & & \vdots \\
   0 & ... & 0 & 0 & 0 & 0 & ... & \sigma_\NN I \\
   \end{pmatrix}.
\end{align*}

Now, we define the Dirichlet-to-Neumann operators $\mathcal{J}_\NN: \mathbb{H}_D \to \mathbb{H}_N$, and  $\Phi: H^1_{loc}(\Omega_0) \to \mathbb{H}_N$ as
\begin{align}
\label{D-to-N}
    \mathcal{J}_\NN(\mathbf{v}) :=\Sconduc \OOM_\NN^{-1} \begin{pmatrix} - \IHDH \vvector \\[2mm] \XXX_{\NN} \IHDH \vvector	\end{pmatrix}, \mbox{ and }  \Phi(\phi_e) := \Sconduc \OOM_\NN \begin{pmatrix} -\trazas^0_{\phi_e} \\[2mm] \XXX_{\NN} \trazas^0_{\phi_e} \end{pmatrix} .
\end{align}

\begin{theorem}[Lemma 4.3 in \cite{HJH18}]
    The operator $\mathcal{J}_\NN: \mathbb{H}_D \to \mathbb{H}_N$ is continuous and coercive.
\end{theorem}

Now we can finally rewrite\footnote{
    The MTF \eqref{operator-M} is almost equal to the one in \cite{HJH18} and \cite{HJA17}. Specifically, \eqref{operator-M} is multiplied by a factor two and the first row does not have a factor $\sigma_j$ as in \cite{HJA17} and \cite{HJH18}.
} Problem \ref{dynamic-problemEP} as an abstract parabolic equation on $\Gamma_j$.
\begin{problem}\label{dynamic-problemEPBIE}
    Let us assume the following as given: a final time $T\in \IR^+$, an external potential $\phi_e \in C([0,T],H^1_{loc}\left( \Omega_0\right))$, and initial conditions $v_j(0)= v_0 \in H^{\frac{1}{2}}(\Gamma_j)$, $\Zz_{j}(0) = \Zz_{j}^0 \in H^{\frac{1}{2}}(\Gamma_j)$, for $j \in \{1,...,\NN\}$. We seek $\mathbf{v}= (v_1, \hdots , v_\NN )^t$, with $v_j \in C([0,T],H^{\frac{1}{2}}(\Gamma_j))$, and $\mathbf{\Zz}= (\Zz_1, \hdots , \Zz_\NN )^t$, $\Zz_{j}\in C([0,T],H^{\frac{1}{2}}(\Gamma_j))$, for $j \in \{1,...,\NN\}$, such that \begin{align}
        \label{ode-principal-vectorial}
    	{\mathbf{C_{m}}} \partial_t \mathbf{v} = - \mathbf{I}^{{ep}}(\mathbf{v} , \mathbf{\Zz} ) - \mathcal{J}_\NN(\mathbf{v}) - \Phi(\phi_e) & \mbox{ on } [0,T]\times \Gamma_j ,
    \end{align}
    where $\mathbf{C_m}$ is a diagonal matrix ${\rm diag}(c_{m,1}, \hdots , c_{m,\NN})$; the operators $\mathcal{J}_\NN(\mathbf{v})$ and $\Phi(\phi_e)$ are defined in \eqref{D-to-N}. The vector
    $\mathbf{I}^{{ep}}(\mathbf{v}, \mathbf{\Zz} ) = (I_1^{{ep}}(v_1, \Zz_1), \hdots , I_\NN^{{ep}}(v_\NN , \Zz_{\NN}) )^t$ satisfies \eqref{current}, \eqref{ODE} and \eqref{beta-used}.
\end{problem}

\section{Numerical Approximation}
\label{sec:numapp}
In this section we propose a numerical solution of Problem \ref{dynamic-problemEPBIE}. We use a multistep semi-implicit time scheme or implicit-explicit method (IMEX), similar to one used in \cite{HJA17, HJH18} (see Section \ref{time-steps}). For the space discretization, we follow an approach analogous to the two-dimensional case presented in \cite{HJH18} using spherical harmonics. Since we do not work with complex valued functions, we employ real spherical harmonics to approximate boundary unknowns. We recall that our dynamic model leads to poorly regular solutions in time.

\subsection{Multistep semi-implicit time scheme}\label{time-steps}
Let $\mathcal{T}_S:=\{t_s\}_{s=0}^{S}$ denote the uniform partition of the time interval $[0,T]$, with $T\in \IR^+$, and $S \in \IN$, where the time step is $\tau = T/S$, and $t_s = s \tau$. Write
\begin{align*}
    t_{s+\frac{1}{2}}:= t_s + \frac{\tau}{2}, \quad s \in \{0, \hdots , S - 1\},
\end{align*}
for the midstep between $t_s$ and $t_{s+1}$. For a time dependent quantity $\phi(t)$, we write $\phi^{(s)}=\phi(t_s)$, and we define the following quantities:
\begin{align*}
    \phi^{(s+\frac{1}{2})}&:= \phi(t_{s+\frac{1}{2}}), &  \overline{\phi}^{(s+\frac{1}{2})} := \frac{\phi^{(s+1)}+ \phi^{(s)}}{2}, \\ \hat{\phi}^{(s+\frac{1}{2})} &:= \frac{{3}\phi^{(s)}- \phi^{(s-1)}}{2}, & \overline{\partial}\phi^{(s)} := \frac{\phi^{(s+1)}- \phi^{(s)}}{\tau}.
\end{align*}
With these, we approximate in time  \eqref{dynamic-problemEPBIE} and \eqref{ODE} as follows:
\begin{align*}
	\mathbf{C_{m}} \overline{\partial} \mathbf{v}^{(s)} = - \mathbf{I}^{{ep}}\left(\widehat{\mathbf{v}}^{(s+\frac{1}{2})}, \widehat{\mathbf{\Zz}}^{(s+\frac{1}{2})} \right) - \mathcal{J}_\NN\left(\overline{\mathbf{v}}^{(s+\frac{1}{2})}\right) - \Phi(\phi_e^{(s+\frac{1}{2})}),\\
    \overline{\partial}^{(s)} \Zz_j = \max \left( \frac{\beta_j(\widehat{v}_j^{(s+\frac{1}{2})}) - \widehat{\Zz_j}^{(s+\frac{1}{2})}}{\tau_{ep,j}}, \frac{\beta_j(\widehat{v}_j^{(s+\frac{1}{2})}) - \widehat{\Zz_j}^{(s+\frac{1}{2})}}{\tau_{res,j}}\right).
\end{align*}
From these expressions, we can notice that
\begin{itemize}
    \item[(i)] At each iteration, the approximation at $t_{s+1}$ requires two previous steps, $t_{s}$ and $t_{s-1}$, but we only have information about the time $t_0$. Thus, we will estimate the values for the time $t_1$ with a predictor-corrector algorithm introduced later in this Section.
    \item[(ii)] Provided with values for the two previous time steps, unknowns for the next time are obtained in terms of $\overline{\partial} \mathbf{v}^{(s)}$, $\overline{\mathbf{v}}^{(s+\frac{1}{2})}$ and $\overline{\partial}^{(s)}$, which are linear. Nonlinear terms are evaluated with values already known, i.e.~they are explicit terms, unlike the others. Hence, the adjectives semi-implicit or IMEX for the multistep method employed. 
    \item[(iii)] At each time step, the discrete problem to be solved is linear. One could choose time-domain schemes with implicit non-linear parts. However, more information about $\mathbf{I}^{ep}$ may be needed. In contrast, our semi-implicit time only requires us to evaluate the function $\mathbf{I}^{ep}$.
\end{itemize}

The predictor-corrector algorithm {can be found} in detail in \cite[Chapter 13]{Thomee2006}, 
\cite[Algorithm 1]{HJA17}. Set $\mathbf{w}^{(0)} = \mathbf{v}^{(0)}$ and $\mathbf{Q}^{(0)}= \mathbf{\Zz}{(0)}$. Then, proceed as follows:
\begin{itemize}
    \item[(I)] {\it Predictor}. First, construct predictions $\mathbf{w}^{(1)}$ and $\mathbf{Q}^{(1)}$ by solving:
    \begin{align*}
        &\mathbf{C_{m}} \overline{\partial} \mathbf{w}^{(0)}  = - \mathbf{I}^{{ep}}\left(\mathbf{w}^{(0)}, \mathbf{Q}^{(0)} \right) - \mathcal{J}_\NN\left(\overline{\mathbf{w}}^{(\frac{1}{2})}\right)  - \Phi \left(\phi_e^{(\frac{1}{2})} \right),\\
        &\overline{\partial} Q_j^{(0)} =  \max \left( \frac{\beta_j( w_j^{(0)}) -  Q_j^{(0)}}{\tau_{ep,j}}, \frac{\beta_j( w_j^{(0)}) -  Q_j^{(0)}}{\tau_{res,j}}\right) \quad \forall j \in \{1,...,\NN\}.
    \end{align*}
    \item[(II)] {\it Corrector}. Then, correct $\mathbf{w}^{(1)}$ and $\mathbf{Q}^{(1)}$ to obtain final values for $\mathbf{v}^{(1)}$ and $\mathbf{\Zz}^{(1)}$ through:
    \begin{align*}
        &\mathbf{C_{m}} \overline{\partial} \mathbf{v}^{(0)} = - \mathbf{I}^{{ep}}\left(\widehat{\mathbf{w}}^{(\frac{1}{2})}, \widehat{\mathbf{Q}}^{(\frac{1}{2})} \right) - \mathcal{J}_\NN \left(\overline{\mathbf{v}}^{(\frac{1}{2})}\right)  - \Phi \left(\phi_e^{(\frac{1}{2})} \right) ,\\
        &\overline{\partial}^{(0)} \Zz_j = \max \left( \frac{\beta_j(\widehat{w}_j^{(\frac{1}{2})}) - \widehat{Q_j}^{(\frac{1}{2})}}{\tau_{ep,j}}, \frac{\beta_j(\widehat{w}_j^{(\frac{1}{2})}) - \widehat{Q_j}^{(\frac{1}{2})}}{\tau_{res,j}}\right) \quad \forall j \in \{1,...,\NN\} .
    \end{align*}
    Then, from the corrector equations, $\mathbf{v}^{(1)}$ and $\mathbf{\Zz}^{(1)}$ are obtained implicitly. 
\end{itemize}

    \begin{remark}
        The predictor-corrector algorithm is only used for obtaining the first time step, as required by the multistep semi-implicit time scheme. One could think of a predictor-corrector algorithm at all time steps but this would entail a higher computational cost and reassess the theoretical convergence results in \cite{HJA17}.
        \end{remark}

Finally, before tackling the spatial discretization, we recall the following result.
\begin{theorem}\label{teo-time-converge}
\cite[Lemma 7]{HJA17}.
Let $\phi \in C^2([0,T]; L^2(\Gamma_j))$, $j\in \{1,\hdots, \NN\}$ then it holds that
\begin{align*}
    \NORM{\overline{\phi}^{(s+\frac{1}{2})} - \phi^{(s+\frac{1}{2})}}_{L^2(\Gamma_j)} &\leq \frac{\tau^2}{4} \max_{t\in [t_s,t_{s+1}]} \NORM{\partial^2_{t}\phi(t)}_{L^2(\Gamma_j)}, \\
   \NORM{\hat{\phi}^{(s+\frac{1}{2})} - \phi^{(s+\frac{1}{2})}}_{L^2(\Gamma_j)} &\leq \frac{5\tau^2}{16} \max_{t\in [t_{s-1},t_{s+1}]} \NORM{ \partial^2_{t}\phi(t)}_{L^2(\Gamma_j)}.
\end{align*}
\end{theorem}

\subsection{Spatial discretization}
We now spatially discretize Problem \ref{dynamic-problemEPBIE}. We start by introducing real spherical harmonics used as the spatial basis for the Dirichlet and Neumann traces \eqref{space-disc-notation}. Then, we proceed with BIOs discretization (see Theorem \ref{diagonal-forms}). Finally, the multistep semi-explicit time method and the spatial discretization are combined into a fully discrete scheme (see Problem \eqref{dynamic-problem-EP-disc}).
\subsubsection{Spherical coordinates and spherical harmonics}
A vector is written as $\mathbf{r}=\left(r,\varphi,\theta\right)^t$, with $r \in [0,\infty)$, $\varphi \in [0,2\pi)$ and $\theta \in [0,\pi]$, which in Cartesian coordinates is equivalent to $\mathbf{r}=r\left(\sin \theta \cos \varphi,\sin \theta \sin \varphi,\cos \theta\right)^t$
. Spherical harmonics of degree $l$ and order $m$ are defined using spherical coordinates as \cite[Section 2]{WieczorekMeschede2018}, \cite[Example 4.3.33]{FreedenGutting2013}:
\begin{subequations}
		\begin{align}
		 Y_{l,m}\left(\theta,\varphi\right) &:= \sqrt{ (2-\delta_{m,0}) \frac{\left(2l+1\right)\left(l-m\right)!}{4 \pi \left(l+m\right)!}} P_l^{m} \left(\cos\theta\right) \cos m \varphi , \mbox{ and} \label{Armonico1} \\
		Y_{l,-m}\left(\theta,\varphi\right) &:= \sqrt{ (2-\delta_{m,0})\frac{\left(2l+1\right)\left(l-m\right)!}{4 \pi \left(l+m\right)!}} P_l^{m} \left(\cos\theta\right) \sin m \varphi , \label{Armonico2}
	\end{align}
\end{subequations}
with $l\in \IN_0$, $m\in \IZ$ such that $0\leq m\leq l$. If $m=0$, $\delta_{m,0}=1$, and it is zero otherwise. $P_l^m$ are the associated Legendre functions of degree $l$ and  order $m$ defined as
\begin{equation*}
	P_{l}^m\left(x\right) := (-1)^m \left( 1- x^2\right)^{\frac{m}{2}} \frac{d^m}{dx^m}P_l(x), \quad \mbox{with} \quad P_{l}\left(x\right) := \frac{1}{2^ll!}\frac{d^l}{dx^l}(x^2-1)^l.
\end{equation*}

Here, the term $(-1)^m$ is the Condon-Shortley phase factor. Spherical harmonics are dense in $C(\mathbb{S}^2)$, with $\mathbb{S}^2$ the surface of the unit sphere, and form a complete orthonormal system in $L^2(\mathbb{S}^2)$ \cite[Section 7.3 and 7.5]{gallier2020spherical}.

Let be $j\in \{1,...,\NN\}$. We define the reference system $j$ as the one centered at $\mathbf{p_j}$ with the same orientation as the system centered at the origin. Furthermore, we denote by $Y_{l,m,j}$ the spherical harmonic $Y_{l,m}$ centered at the origin of the $j$th reference system. Thus, if $\left( r_j, \varphi_j, \theta_j \right) $ are the vector spherical coordinates of $\mathbf{r_j}$ in the reference system $j$, we have that $Y_{l,m,j}\left(\mathbf{r}_j\right)=Y_{l,m}\left(\theta_j, \varphi_j\right)$.

For $L \in \IN_0$ and $j\in \{1,...,\NN\}$, we define subspaces 
\begin{equation}\label{Y-space}
\mathcal{Y}_L\left(\Gamma_j \right):= \mbox{span}\left\lbrace Y_{l,m,j}: l \in \IN_0, m \in \IZ, l \leq L, |m|\leq l \right\rbrace,
\end{equation}
equipped with the $L^2(\Gamma_j)-$norm. Notice that the dimension of each subspace is $(L+1)^2$, and that the sequence of subspaces $\lbrace \mathcal{Y}_L \left(\Gamma_j \right) \rbrace_{L \in \IN_0} $ is dense in $H^{\frac{1}{2}}(\Gamma_j)$ and in $H^{-\frac{1}{2}}(\Gamma_j)$. The result follows from the density of spherical harmonics in the spaces of continuous functions
. This last result justifies the discretization of all boundary Dirichlet and Neumann unknowns with spherical harmonics. At a given time $t$, for $j \in \lbrace 1, ..., \NN \rbrace$, we write $\uojdl$, $\uojnl$, $\ujdl$, $\ujnl$, $\vjl$ and $\zjl$ in $\mathcal{Y}_L(\Gamma_j)$ for the approximations of $\trazaD^{0j} u_0 $, $\trazaN^{0j} u_0 $, $\trazaD^{j} u_j  $, $\trazaN^{j} u_j $, $v_j$ and $\Zz_j$, respectively. They can be written as the following series expansions:
\begin{subequations}\label{space-disc-notation}
    \begin{align}
    \uojdl&=\sum_{l=0}^{L}  \sum_{m=-l}^l \ulmojdl  Y_{l,m,j}, & \uojnl &=\sum_{l=0}^{L}  \sum_{m=-l}^l \ulmojnl  Y_{l,m,j},\\
    \ujdl&=\sum_{l=0}^{L}  \sum_{m=-l}^l \ulmjdl  Y_{l,m,j},& \ujnl &=\sum_{l=0}^{L}  \sum_{m=-l}^l \ulmjnl  Y_{l,m,j},\\
    \vjl&=\sum_{l=0}^{L}  \sum_{m=-l}^l \vlmjl  Y_{l,m,j}, &  \zjl&=\sum_{l=0}^{L}  \sum_{m=-l}^l \zlmjl  Y_{l,m,j}
    \end{align}
\end{subequations}
with $\ulmojdl$, $\ulmojnl$, $\ulmjdl$, $\ulmjnl$, $\vlmjl$, and $\zlmjl$ being constants in space but varying in time. Notice that the norm in $\mathcal{Y}_L \left(\Gamma_j \right)$ of any of these functions is the square root of the sum of squared coefficients times the radius of $\Gamma_j$, i.e. 
\begin{equation}\label{norm-discrete-spaces}
    \NORM{\vjl}_{\mathcal{Y}_L \left(\Gamma_j \right)}^2 = \radio_j \sum_{l=0}^L \sum_{m=-l}^l (\vlmjl)^2.
\end{equation}

Finally, let $\mathbb{Y}_L:= \Pi_{j=1}^{\NN} \mathcal{Y_L}(\Gamma_j)$, and define the following vectors in $\mathbb{Y}_L$: 
\begin{subequations}\label{space-disc-notation-2}
\begin{align}
    \mathbf{v}^L &:= \begin{pmatrix}
    v_1^L, \hdots, \vjl, \hdots, v_{\NN}^L \end{pmatrix}^t, &\mathbf{Z}^L &:= \begin{pmatrix}
    z_1^L, \hdots, \zjl, \hdots, z_{\NN}^L
    \end{pmatrix}^t, \\
    \mathbf{u}_{D,0}^L &:= \begin{pmatrix}
    u_{D,01}^{L}, \hdots, \uojdl, \hdots, u_{D,0\NN}^{L}
    \end{pmatrix}^t, &\mathbf{u}_{D}^L &:= \begin{pmatrix}
    u_{D,1}^{L}, \hdots, \ujdl , \hdots,    u_{D,\NN}^{L}
    \end{pmatrix}^t,\\ 
    \mathbf{u}_{N,0}^L &:= \begin{pmatrix}
    u_{N,01}^{L}, \hdots, \uojnl, \hdots,   u_{N,0\NN}^{L}
    \end{pmatrix}^t, & \mathbf{u}^L_D &:= \begin{pmatrix}
    u_{N,1}^{L}, \hdots ,  \ujnl , \hdots,  u_{N,\NN}^{L} \end{pmatrix}^t.
\end{align}
\end{subequations}

The norm for a function in $\mathbb{Y}_L$, for example, $\mathbf{v}^L${,} is $\displaystyle \NORM{\mathbf{v}^L}_{\mathbb{Y}_L}^2 = \sum_{j=1}^{\NN}||\vjl||_{\mathcal{Y}_L \left(\Gamma_j \right)}^2.$

\subsubsection{BIOs discretization}\label{BIOs-discretization}

The fundamental solution can be expanded using spherical harmonics \cite[Theorem 4.3.29, Lemma 4.4.1 and Remark 4.4.2]{FreedenGutting2013} as the following result shows.
\begin{theorem}\label{TeoSeriesFundamental}
Let $\mathbf{r}$, $\mathbf{r}'$ be vectors, whose spherical coordinates in the reference system $j$ are $\left(r_j,\theta_j, \varphi_j \right)$ and $\left(r_j',\theta_j', \varphi_j' \right)$, respectively. For $r_j>r_j'$ we have that
\begin{equation}\label{FundamentalSeries}
g\left(\mathbf{r}, \mathbf{r}'\right) = \sum_{l=0}^{\infty}  \frac{1}{2l+1} \frac{r_j^{'l}}{r_j^{l+1}}\sum_{m=-l}^l Y_{l,m,j}\left(\mathbf{r}\right)  Y_{l,m,j}\left(\mathbf{r}'\right).
\end{equation}
Moreover, the series (\ref{FundamentalSeries}) and its term by term first derivatives with respect to $r_j$ or $r_j'$ are absolutely and uniformly convergent on compact subsets with $r_j>r_j'$ \cite[Section 2.3, p.23 and p.24]{ColtonKress2013a}. 
\end{theorem}

\begin{theorem}\label{diagonal-forms}
The diagonal forms of the BIOs \eqref{BIOS-definition} are:
\begin{align*}
	V_{j,j}^0\left( Y_{l,m,j} \right) &= \frac{1}{2l+1} \radio_j Y_{l,m,j}, & V_{j}\left( Y_{l,m,j} \right)&=  \frac{1}{2l+1} \radio_j Y_{l,m,j},  \\
	K_{j,j}^0\left( Y_{l,m,j} \right) &= \frac{1}{2(2l+1)}Y_{l,m,j}, &	K_{j}\left( Y_{l,m,j} \right) &= -\frac{1}{2(2l+1)} Y_{l,m,j},\\
	K^{*0}_{j,j}\left( Y_{l,m,j} \right) &= \frac{1}{2l+1} Y_{l,m,j}, & K^*_{j}\left( Y_{l,m,j} \right) &= -\frac{1}{2(2l+1)} Y_{l,m,j},\\
	W_{j,j}^0\left( Y_{l,m,j} \right) &=  \frac{l(l+1)}{2l+1} \frac{1}{\radio_j} Y_{l,m,j}, & 	W_{j}\left( Y_{l,m,j} \right)&=  \frac{l(l+1)}{2l+1} \frac{1}{\radio_j}  Y_{l,m,j}.
\end{align*}
\end{theorem}

\begin{proof}
The result follows from Theorem \ref{TeoSeriesFundamental}, the orthonormality of spherical harmonics, and the definitions of the BIOs presented in \eqref{BIOS-definition}. Similar diagonal forms can also be found in \cite[Section 3 and Table 2]{VicoGreengardGimbutas2014}, where the result is stated for complex spherical harmonics on the unit sphere.
\end{proof}

\begin{corollary}\label{disc-diagonals}
The following holds
\begin{align*}
	\left( V_{j,j}^0\left( Y_{l,m,j} \right) , Y_{p,q,j} \right)_{L^2(\Gamma_j)} &=  \left( V_{j}\left( Y_{l,m,j} \right) , Y_{p,q,j} \right)_{L^2(\Gamma_j)} =\frac{\radio_j^3 }{2l+1} \delta_{l,p} \delta_{m,q} ,  \\
	\left( K_{j,j}^0\left( Y_{l,m,j} \right) , Y_{p,q,j} \right)_{L^2(\Gamma_j)}  &= -\left( K_{j}\left( Y_{l,m,j} \right) , Y_{p,q,j} \right)_{L^2(\Gamma_j)} =\frac{\radio_j^2}{2(2l+1)} \delta_{l,p} \delta_{m,q}, \\
	\left( K^{*0}_{j,j}\left( Y_{l,m,j} \right) , Y_{p,q,j} \right)_{L^2(\Gamma_j)}  &= -\left( K^*_{j}\left( Y_{l,m,j} \right)  , Y_{p,q,j} \right)_{L^2(\Gamma_j)} = \frac{\radio_j^2}{2(2l+1)} \delta_{l,p} \delta_{m,q},   \\
	\left( W_{j,j}^0\left( Y_{l,m,j} \right)  , Y_{p,q,j} \right)_{L^2(\Gamma_j)}  &=  \left(	W_{j}\left( Y_{l,m,j} \right)  , Y_{p,q,j} \right)_{L^2(\Gamma_s)} =  \frac{l(l+1)}{2l+1} \radio_j \delta_{l,p} \delta_{m,q},
\end{align*}
with $\delta_{l,p}$, $\delta_{m,q}$ denoting the standard Kronecker deltas. Also, for the scalar identity operators presented in Section \ref{sec:BIOs}, it holds that $\displaystyle \left( I \left( Y_{l,m,j} \right)  , Y_{p,q,j} \right)_{L^2(\Gamma_j)} = \radio_j^2 \delta_{l,p} \delta_{m,q}$.

\end{corollary}

Cross-interaction operators, e.g.~$V_{i,j}^0$ for  $i \not = j$, are non-singular and generally non diagonalizable. The double and single layer operators analytic expressions can be used to compute  the non-singular integrals for $i \not = j$: 
\begin{subequations}\label{cross-interactions}
\begin{align}
	\left( V_{i,j}^0\left( Y_{l,m,j} \right) ; Y_{p,q,i} \right)_{L^2(\Gamma_i)} &= \int_{\Gamma_i} SL_{0j}(Y_{l,m,j}) Y_{p,q,i} \ d \Gamma_i , \label{cross-v}\\
	\left( K_{i,j}^0\left( Y_{l,m,j} \right) ; Y_{p,q,i} \right)_{L^2(\Gamma_i)} &= \int_{\Gamma_i} DL_{0j}(Y_{l,m,j}) Y_{p,q,i}  \ d \Gamma_i,\\
	\left( K^{*0}_{i,j}\left( Y_{l,m,j} \right) ; Y_{p,q,i} \right)_{L^2(\Gamma_i)} &= \int_{\Gamma_i}  \widehat{\mathbf{n}}_{0i} \cdot\nabla SL_{0j}(Y_{l,m,j})  Y_{p,q,i} \ d \Gamma_i,\\
	\left( W_{i,j}^0\left( Y_{l,m,j} \right)  ; Y_{p,q,i} \right)_{L^2(\Gamma_i)} &= - \int_{\Gamma_i}   \widehat{\mathbf{n}}_{0i}\cdot \nabla DL_{0j}(Y_{l,m,j}) Y_{p,q,i} \ d \Gamma_i .
\end{align}
\end{subequations}

Approximations of the integrals \eqref{cross-interactions} are provided via Gauss-Legendre quadratures. Specifically, along $\theta$, we use the change of variable $u=\cos\theta$. Then, variable functions are sampled at the zeros of the Legendre Polynomial of degree $L_c+1$, whereas the trapezoidal rule is applied to equally spaced nodes in $\varphi$, with $2L_c+1$ points. If the function being integrated has a spherical harmonic expansion with coefficients equal to zero for degrees higher than $L_c$, then the quadrature yields the exact result, assuming that there are not other sources of error \cite{WieczorekMeschede2018}. Moreover, quadrature in $\varphi$ can be computed using the Fast Fourier Transform.

\begin{remark}
One would expect $L_c$ to be greater than $p$ and $l$ in \eqref{cross-interactions}. Yet, without further analysis, it is not known if a polynomial of degree $L_c$ is a good approximation for $SL_{0j}(Y_{l,m,j}) Y_{p,q,i}$, $DL_{0j}(Y_{l,m,j})Y_{p,q,i}$, $\nabla SL_{0j}(Y_{l,m,j}) \cdot \widehat{\mathbf{n}}_{0i}\ Y_{p,q,i}$ and $\nabla SL_{0j}(Y_{l,m,j}) \cdot \widehat{\mathbf{n}}_{0i} \ Y_{p,q,i}$, since, as the translation theorems for spherical harmonics highlight, the translation of only one spherical harmonic is expressed as another infinite series of spherical harmonics. Also, notice that \eqref{cross-interactions} can also be computed using a translation theorem for real spherical harmonics as in \cite{AganinDavletshin2018}. In this case, the integral has an explicit expression and does not need to be computed numerically. Instead, the computing efforts focus on calculating the coefficients given by the translation theorem.
\end{remark}

\begin{corollary}
The following holds
\begin{align*}
\left( V_{i,j}^0\left( Y_{l,m,j} \right) ; Y_{p,q,i} \right)_{L^2(\Gamma_i)} &= \left( V_{j,i}^0\left( Y_{p,q,i} \right) ; Y_{l,m,j} \right)_{L^2(\Gamma_j)},\\
\left( K_{i,j}^0\left( Y_{l,m,j} \right) ; Y_{p,q,i} \right)_{L^2(\Gamma_i)} &= - \frac{l}{\radio_j} \left( V_{i,j}^0\left( Y_{l,m,j} \right) ; Y_{p,q,i} \right)_{L^2(\Gamma_i)}, \\
\left( K_{j,i}^{*0}\left( Y_{p,q,i} \right) ; Y_{l,m,j} \right)_{L^2(\Gamma_j)} &= \left( K_{i,j}^0\left( Y_{l,m,j} \right) ; Y_{p,q,i} \right)_{L^2(\Gamma_i)}, \\
\left( W_{i,j}^0\left( Y_{l,m,j} \right) ; Y_{p,q,i} \right)_{L^2(\Gamma_i)} &= \frac{l}{\radio_j} \left( K_{i,j}^{*0}\left( Y_{l,m,j} \right) ; Y_{p,q,i} \right)_{L^2(\Gamma_i)} .
\end{align*}
\end{corollary}
\begin{proof}
The result follows from Theorem \ref{TeoSeriesFundamental}, the orthonormality of spherical harmonics along with the definition of the BIOs.
\end{proof}

From this last corollary, it can be deduced that the integrals of all the cross-interactions of a couple of spheres $i$ and $j$ \eqref{cross-interactions} can be derived having the results of the expression \eqref{cross-v} for all of the $l$, $m$, $p$ and $q$ needed, which avoids the need of computing numerically the other integral expressions.

\subsection{Fully discrete scheme}

Following Section \ref{time-steps}, we state the multistep semi-implicit in time and space numerical discretization of Problem \ref{dynamic-problemEPBIE}:
\begin{problem}\label{dynamic-problem-EP-disc}
    Let $\mathbf{v}^{L,(0)}$ and $\mathbf{\Zz}^{L,(0)}$ in $\mathbb{Y}_L$ be given. Then, for $s \in \{ 2, ..., S-1 \}$, we seek $\mathbf{v}^{L,(s)}$, $\mathbf{\Zz}^{L,(s)}$ in $\mathbb{Y}_L$ solution of:
        \begin{align}
            \left( \mathbf{C_m} \overline{\partial} \mathbf{v}^{L,(s)}  + \mathcal{J}_{\NN} \left(\overline{\mathbf{v}}^{L,(s+\frac{1}{2})}\right) + \mathbf{I}^{{ep}} \left( \hat{ \mathbf{v} }^{ L, (s+\frac{1}{2}) }, \hat{\mathbf{\Zz}}^{L, (s+\frac{1}{2})} \right) + \Phi \left(\phi_e^{(s+\frac{1}{2})} \right) , \mathbf{y}\right)_{\mathbb{Y}_L} =0\label{discret-1}  
        \end{align}
    \begin{align}
            \overline{\partial}^{(s)} \Zz_j^L =& \max \left( \frac{\beta_j(\widehat{v}_j^{L, (s+\frac{1}{2})}) - \widehat{\Zz_j}^{L, (s+\frac{1}{2})}}{\tau_{ep,j}}, \frac{\beta_j(\widehat{v}_j^{L, (s+\frac{1}{2})}) - \widehat{\Zz_j}^{L, (s+\frac{1}{2})}}{\tau_{res,j}}\right),
        \end{align}
    for all $\mathbf{y} \in \mathbb{Y}_L$. For $s=1$ we use the equivalent weak formulation of the corrector-predictor algorithm presented in \ref{time-steps}.
\end{problem}

In order to solve Problem \ref{dynamic-problem-EP-disc}, at each time step, with the exception of the predictor-corrector algorithm, we solve the weak linear system equivalent to
\begin{align}\label{discrete-system}
\left[\begin{smallmatrix}
	4 \OOA_{0,\NN} & -2\XXX_{\NN}^{-1} & \mathbf{I}_{2\NN \times \NN} \\
	 -2\XXX_{\NN}& 4 \OOA_{1,\NN} & - \XXX_{\NN} \mathbf{I}_{2\NN \times \NN} \\
	  & \boldsymbol{\sigma}_{\NN \times 4\NN} & \frac{1}{\tau} \mathbf{C_m} \\
\end{smallmatrix}\right] \left( \begin{smallmatrix}
	\overline{\mathbf{u}}_{D,0}^{L,(s+1/2)} \\
	\overline{\mathbf{u}}_{N,0}^{L,(s+1/2)} \\
	\overline{\mathbf{u}}_{D}^{L,(s+1/2)} \\
	\overline{\mathbf{u}}_{N}^{L,(s+1/2)} \\
	\vvector^{L,(s+1)}
\end{smallmatrix} \right) = \left( \begin{smallmatrix}
-\left( 2 \trazas^0_{\phi_e^{L,(s+\frac{1}{2})}} + \mathbf{I}_{2\NN \times \NN} \vvector^{L,(s)}  \right)\\ 
\XXX_{\NN} \left( 2 \trazas^0_{\phi_e^{L,(s+\frac{1}{2})}} +\mathbf{I}_{2\NN \times \NN} \vvector^{L,(s)} \right) \\
\frac{1}{\tau} \mathbf{C_m} \vvector^{L,(s)} - \mathbf{I}^{{ep}} \left( \hat{ \vvector }^{L, (s+\frac{1}{2}) }, \hat{\mathbf{\Zz}}^{L, (s+\frac{1}{2})} \right)
\end{smallmatrix} \right),
\end{align} 
where the test function is in $\mathbb{Y}_L \times\mathbb{Y}_L \times\mathbb{Y}_L \times\mathbb{Y}_L \times\mathbb{Y}_L$. Notice that we obtain  mid-steps ($s+1/2$) for traces of  extra- and intracellular potentials, whereas only the transmembrane potential is obtained at time steps $s$.

\begin{remark}
    With the exception of the scalar operators inside of $\OOA_{0,\NN}$ and $\mathbf{I}^{ep}$, which are computed numerically, all other matrices are diagonalizable and analytic for the geometry here considered (Theorem \ref{disc-diagonals}). Thus, the discrete matrix used to solve at each time step is almost entirely block diagonal. Note that if changing $\mathbf{I}^{ep}$ without modifying the dynamics for the transmembrane potentials, leads to a modified right-hand side  in the linear system of equation \eqref{discrete-system}.
\end{remark}

\begin{remark}
    The time step needs to be bounded by the smallest characteristic time of the system to ensure stability. In the work \cite{HJH18}, the bound is independent of the spatial discretization but depends on the parameters of the non-linear problem. Moreover, given the poor regularity of functions $Z_j$, we cannot guarantee high-order convergence in time. Yet, the use of spherical harmonics in space greatly reduces the overall number of degrees of freedom and leads to better convergence rates than first-order boundary element discretizations \cite{stephan_convergence_1989}.
\end{remark}

\section{Numerical Results}
\label{sec:numres}
In this section, we verify and test the proposed computational scheme. To this end, we first check the MTF implementation for single and multiple cells to then combine it with the multistep semi-implicit time-domain method. Next, we perform tests for linear and non-linear dynamics. Physical parameters used throughout are given in \cite[Table 1]{MistaniGuittetea2019} and \cite[Table 1]{KavianLeguebeea2014}.
\subsection{Hardware and Code Implementation}

Numerical results were obtained in a machine with Quad Core Intel Core i7-4770 (-MT MCP-), 1498 MHz, 31'982.1 MiB RAM (90\% available for computations), with operating system Linux Mint 20.3 Una and Kernel: 5.4.0-131- generic x86\_64. Simulation codes were programmed on Python 3.10. Its installation was achieved via the open-source platform Anaconda\footnote{\url{https://www.anaconda.com/products/distribution}}, Conda\footnote{\url{https://docs.conda.io/projects/conda/en/stable/}} 4.13.0, and using the conda-forge repository.\footnote{The following packages were installed explicitly: \texttt{pyshtools} 4.10 \cite{WieczorekMeschede2018}, (conda install pyshtools=4.10), \texttt{numpy} 1.23.1, \texttt{scipy} 1.9.0, and \texttt{matplotlib-base} 3.5.2.} With the \texttt{numpy} library, we take advantage of vectorized computations. Moreover, we only use direct solvers, without any parallelization or matrix compression, which of course can be performed. 
\subsection{Code validation}
In order to validate our code, we check that computed solutions fulfill discrete Calder\'on identities at the boundaries as well as discrete jump conditions. Being approximations, these properties do not hold exactly, thus we define the following errors:
\begin{itemize}
    \item Discrete Calder\'on exterior and interior errors respectively:
\begin{align}\label{def-cal-ex}
    \NORM{  \left( 2 \OOA_{0, \NN} - \mathbf{I} \right) \left( \begin{smallmatrix}
        \mathbf{u}_{D,0}^{L,(s+1)} \\
            \mathbf{u}_{N,0}^{L,(s+1)} \\
    \end{smallmatrix} \right)}_{\mathbb{Y}_L \times \mathbb{Y}_L}, 
    \NORM{  \left( 2 \OOA_{1, \NN} - \mathbf{I} \right) \left( \begin{smallmatrix}
        \mathbf{u}_{D}^{L,(s+1)} \\
            \mathbf{u}_{N}^{L,(s+1)} \\
    \end{smallmatrix} \right) }_{\mathbb{Y}_L \times \mathbb{Y}_L}.
\end{align}
\item 
Jump error:
\begin{align}\label{def-jump-err}
    \NORM{ \left( \begin{smallmatrix}
        \mathbf{u}_{D,0}^{L,(s+1)} \\
            \mathbf{u}_{N,0}^{L,(s+1)} \\
    \end{smallmatrix} \right) - \mathbf{X}_{\NN}^{-1}  \left( \begin{smallmatrix}
        \mathbf{u}_{D}^{L,(s+1)} \\
            \mathbf{u}_{N}^{L,(s+1)} \\
    \end{smallmatrix} \right) + \mathbf{I}_{2\NN \times \NN }\vvector^L + {\trazas}^{0j} \phi_e^L }_{\mathbb{Y}_L \times \mathbb{Y}_L} \approx 0.
\end{align}
\end{itemize}

Here the norm $\|\cdot\|_{\mathbb{Y}_L \times \mathbb{Y}_L}$ is computed as
\begin{equation*}
    \NORM{\left(\begin{smallmatrix}
        \mathbf{u}_{D}^{L,(s+1)} \\
            \mathbf{u}_{N}^{L,(s+1)} \\
    \end{smallmatrix} \right)}_{\mathbb{Y}_L \times \mathbb{Y}_L}^2 = \NORM{\mathbf{u}_{D}^{L,(s+1)}}^2_{\mathbb{Y}_L} + \NORM{\mathbf{u}_{N}^{L,(s+1)}}^2_{\mathbb{Y}_L} .
\end{equation*}

In what follows, we will use the following notations:
\begin{itemize}
    \item Relative error in $L^2(\Gamma_j)$:
    \begin{align}\label{re_{2}}
        re_{2}(\phi_1, \phi_2)_j:= \frac{\NORM{\phi_1 - \phi_2}_{L^2(\Gamma_j)}}{\NORM{\phi_1}_{L^2(\Gamma_j)}}.
    \end{align}
This error is computed for spherical harmonics expansions when possible \eqref{norm-discrete-spaces} or using the numerical quadrature presented at the end of Section \ref{BIOs-discretization}.
    \item Relative error in $C^0 \left( (0,T), L^2(\Gamma_1) \right)$:
    \begin{align}\label{REinf2}
        re_{\infty,2}(\phi_1, \phi_2)_j:= \frac{\max_{t_s \in T_s}\NORM{\phi_1(t_s+\tau/2) - \phi_2(t_s+\tau/2)}_{L^2(\Gamma_j)}}{\max_{t_s \in T_s}\NORM{\phi_1(t_s+\tau/2)}_{L^2(\Gamma_j)}}.
    \end{align}
    \item Relative error in $L^2 \left( (0,T), L^2(\Gamma_1) \right)$:
    \begin{align}\label{re_{2}2}
        re_{2,2}(\phi_1, \phi_2)_j:= \frac{\NORM{\phi_1 - \phi_2}_{L^2 \left( (0,T), L^2(\Gamma_1) \right)}}{\NORM{\phi_1}_{L^2 \left( (0,T), L^2(\Gamma_1) \right)}}.
    \end{align}
    The approximation of the time integral is done by a composite trapezoidal rule using the points of the computed time mid-steps.
\end{itemize}

\subsubsection{MTF Validation}\label{sec:mtf-veri}
We verify first the MTF method without time evolution, by solving \eqref{MTF1} for four different geometrical configurations and sources. In all four experiments, we set $\vvector = \textbf{0}$ and use the point source function $\displaystyle \phi_{e} = 1/ (4\pi \sigma_0 \NORM{\mathbf{r} - \mathbf{p_0}}_2)$ as the external applied potential.

\begin{itemize}
\item \textbf{Example 1}: One sphere centered at the origin with intracellular conductivity $\sigma_1$ different from $\sigma_0$.
\item \textbf{Example 2}: Three (aligned) spheres. The first and the third one have conductivity $\sigma_0$ (phantom spheres), while the one in the middle has a different conductivity $\sigma_1$. 
\end{itemize}

The parameters used for validation for Examples 1 and 2 for a single sphere are presented in Table \ref{space-convergence-parameters}, additional parameters for Example 2 are presented in Table \ref{space-convergence-2spheres-parameters}. In Example 1, the sphere has a different conductivity than the extracellular space, and an analytic solution can be obtained. In Figure \ref{mtf-results-converge-1-sphere-point-source} the relative errors in $L^2(\Gamma_1)$ \eqref{re_{2}} of the computed solutions for different $L$ against the analytic solution are presented. The image shows the expected exponential convergence with respect to the maximum degree of the spectral basis $L$. Example 2 involves three spheres, two of those having the same properties as the external medium, while the one in the middle is different (see Figure \ref{3-analytic-result}). {Therefore}, the traces of the latter should be equal to the ones computed without the first two, i.e.~the same as in Example 1. The relative $L^2(\Gamma_1)$ error of the difference between the analytic solution for the four traces and the numerical one corresponding to the sphere with different conductivity, is $6.06 \cdot 10^{-15}$. In Figure \ref{3-analytic-result}, $u_0^{50}$ is plotted where the only sphere showing a response to $\phi_{e}$ is the sphere in the middle that has different properties compared to the external medium. Discrete Calder\'on and jump errors are of order $10^{-16}$ considered as zero.


\begin{table}[tp]
\footnotesize
\caption{Parameters used in Section \ref{sec:mtf-veri} for Example 1 and 2 for the MTF validation. Conductivity values are from \protect\cite[Table 1]{KavianLeguebeea2014}, cell radius from \cite[Table 1]{MistaniGuittetea2019}. }
        \label{space-convergence-parameters}
    \begin{center}
        \begin{tabular}{|l|l|l|l|}
        \hline
        Parameter & Symbol      & Example 1       & Unit \\
        \hline
        Intensity & $a$             &    1   & $\mu$A \\
        Source position & $\mathbf{p_0}$ &  (0, 0, 20)   &  $\mu$m     \\
        Extracellular conductivity &$\sigma_0$   & 5        & $\mu$S/$\mu$m     \\
        Intracellular conductivity &$\sigma_1$  &  0.455         & $\mu$S/$\mu$m \\
        Cell radius & $\radio_1$   &10 &  $\mu$m \\
        Maximum degree of spherical harmonics & $L$ &50 & \\
        \hline
        \end{tabular}
    \end{center}
\end{table}


\begin{figure}[tp]
	\begin{center}
		\includegraphics[width=0.6\textwidth]{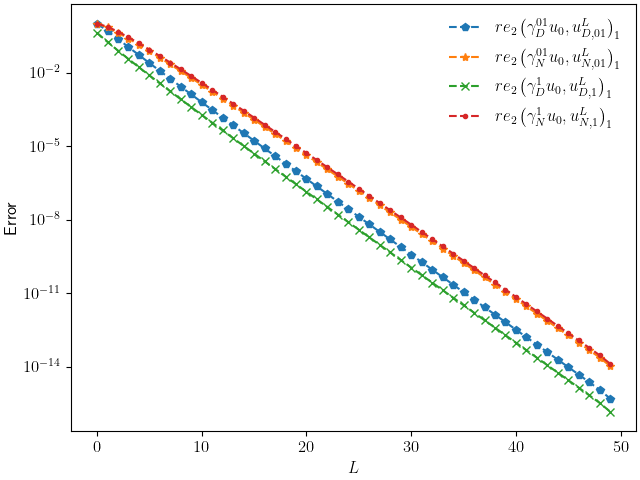}
		\caption{Error convergence for traces in Example 1 (Section \ref{sec:mtf-veri}). The relative error $L^2(\Gamma_1)$ \eqref{re_{2}} is computed against the analytic solution with parameter values in Table \ref{space-convergence-parameters}.}
		\label{mtf-results-converge-1-sphere-point-source}
	\end{center}
\end{figure}
\begin{table}[tp]
\footnotesize
\caption{Parameters used for the MTF verification with $\phi_{e} = 1/ (4\pi \sigma_0 \NORM{\mathbf{r} - \mathbf{p_0}}_2)$ in Example 2, Section \ref{sec:mtf-veri}. Conductivities are given in \protect\cite[Table 1]{KavianLeguebeea2014} and radii are in \protect\cite[Table 1]{MistaniGuittetea2019}.}
        \label{space-convergence-2spheres-parameters}
    \begin{center}
        \begin{tabular}{|l|l|l|l|}
        \hline 
        Parameter    & Symbol   & Value         & Unit \\
        \hline
        Cell 1 intracellular conductivity & $\sigma_1$      & 0.455        & $\mu$S/$\mu$m \\
        Cell 2 and 3 intracellular conductivity & $\sigma_2$, $\sigma_3$      & 5     &  $\mu$S/$\mu$m \\
        Cell 1 radius & $\radio_1$ & 10 &  $\mu$m \\
        Cell 2 radius & $\radio_2$ & 8 &  $\mu$m \\
        Cell 3 radius & $\radio_3$ & 9 &  $\mu$m \\
        Cell 1 center position & $\mathbf{p_1}$  &  (0, 0, 0)   &  $\mu$m     \\
        Cell 2 center position & $\mathbf{p_2}$  &  (25, 0, 0)   &  $\mu$m     \\
        Cell 3 center position & $\mathbf{p_3}$  &  (-24, 0, 0)   &  $\mu$m     \\
        Maximum degree of spherical harmonics & $L$ & 50 &   \\
        Quadrature degree & $L_c$ & 100 &  \\
        \hline
        \end{tabular}
    \end{center}
\end{table}
\begin{figure}[tp]
    \centering
    \subfloat[Plane $y=0$.]{
    \includegraphics[width=0.49\textwidth]{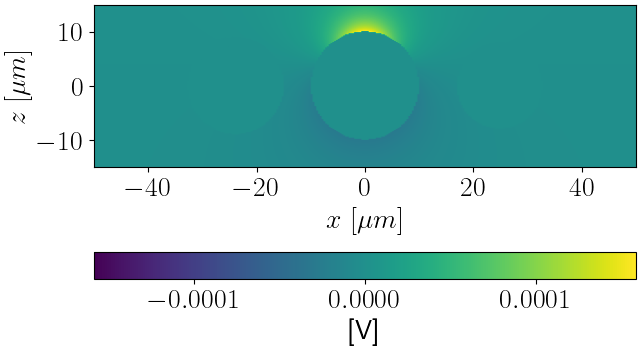}}
    \subfloat[Plane $z=0$.]{
    \includegraphics[width=0.49\textwidth]{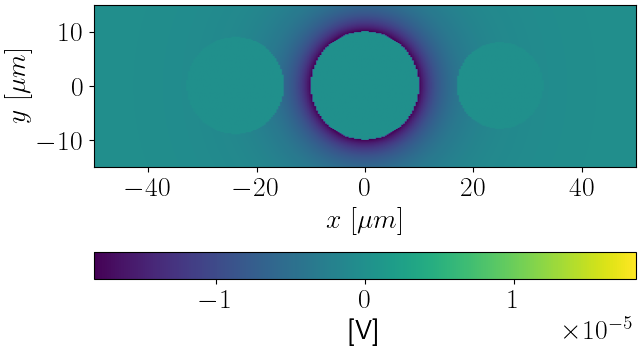}}
	\caption{Field $u_0^{50}$ of Example 2, Section \ref{sec:mtf-veri} with parameters from Table \ref{space-convergence-2spheres-parameters}.}
	\label{3-analytic-result}
\end{figure}

\subsubsection{Multistep semi-implicit time approximation validation: linear case}\label{sec:semiimpveri}
We validate the proposed time scheme by solving problem \eqref{discret-1} for a linear current with only one cell $$\displaystyle c_{m,1} \partial_t v_1 + \frac{1}{r_{m,1}} v_j = - \sigma_1 \trazaN^1 u_1,$$
where instead of $I_1^{ep}(v_1, Z_1)$ we use $r_{m,1}^{-1} v_j$. Additionally, we assume that $\phi_e$ can be factorized $\phi_e(t,\mathbf{r}) = \phi_{time}(t) \phi_{space}(\mathbf{r})$. If $\phi_{space}$ is expanded in spherical harmonics, the coefficients for the equivalent expansion of $v_1$, denoted by $v_1^{l,m}$, can be obtained by solving $$\displaystyle \partial_tv_1^{l,m} + \alpha_1^{l,m} \ v_1^{l,m} = - \beta_1^{l,m} \ \phi_{time}(t),$$ with
\begin{align*}
\alpha_1^{l,m} :=  \frac{1}{c_m R_m} + \frac{\sigma_0 \sigma_1 l (l+1) }{c_m \radio_1 ( \sigma_0(l+1) + \sigma_1 l )}, \ 
\beta_1^{l,m} := \frac{\sigma_0 \sigma_1 l (b_{d,l,m}(l+1) - b_{n,l,m} \radio_1 ) }{c_m \radio_1( \sigma_0(l+1) + \sigma_1 l ) },
\end{align*}
where $b_{d,l,m}$ and $b_{n,l,m}$ are the coefficients of degree $l$ and order $m$ for the Dirichlet and Neumann trace expansions of $\phi_{space}$ on the cell's membrane, respectively. Then, the spherical harmonic expansion coefficients of $v_1$ are 
\begin{equation*}
    v_1^{l,m}(t) = - \beta_1^{l,m} e^{- \alpha_1^{l,m} t} \int_0^t \phi_{time}(s)e^{ \alpha_1^{l,m} s} ds + v_1^{l,m}(0) e^{- \alpha_1^{l,m} t}.
\end{equation*}
We present simulation results for two different time behaviors for $\phi_e$, $\phi_{time-exp} = e^{-t}$ and $\phi_{time-cte} = 1$. We use a point source function for the spatial part of $\phi_e$. Parameters are presented in Table \ref{time-valitadion-parameters}. In Figure \ref{time-verification-results-01}, the absolute error of the difference between $\overline{v}_1^{23}$ ($\tau=2.5\cdot 10^{-2} \mu$s) and $v_1$ in space is presented for each mid-time step. We compute also $\frac{\tau^2}{4} \NORM{\partial^2_{t}v_1}_{L^2(\Gamma_j)}$ to validate the first bound in Theorem \ref{teo-time-converge}. For $\phi_{time-exp}$, the absolute error satisfies the first bound in Theorem \ref{teo-time-converge} everywhere except for the range between $0.4 \ \mu$s and $0.7 \ \mu$s, where a dip in $\frac{\tau^2}{4} \|\partial^2_{t}v_1(t)\|_{L^2(\Gamma_j)}$ occurs. This is due to the second derivative of the most significant term ($l=1$) approaching zero (see Figure \ref{plotderivative}). In contrast, for $\phi_{time-cte}$ the bound is fulfilled at all times.

\begin{table}[tp]
\footnotesize
\caption{Parameters used for the time scheme validation in Section \ref{sec:semiimpveri} where linear dynamics are assumed. The external potential is $\phi_e = I(t)/ (4\pi \sigma_0 \NORM{\mathbf{r} - \mathbf{p_0}}_2)$ and only one cell is considered. Conductivity values are given in \protect\cite[Table 1]{KavianLeguebeea2014}, the cell radius and the specific membrane capacitance are given in  \protect\cite[Table 1]{MistaniGuittetea2019}, and the specific membrane resistance is from \protect\cite[Table 1]{HJA17}. }
        \label{time-valitadion-parameters}
    \begin{center}
        \begin{tabular}{|l|l|l|l|}
        \hline
        Parameter & Symbol      & Values         & Unit \\
        \hline
        Intensity & $I(t)$             &    $e^{-t}$ and $1$    &   $\mu$A \\
        Source position & $\mathbf{p_0}$ & (0, 0, 50)  &    $\mu$m     \\
        Extracellular conductivity &$\sigma_0$   & 5.00  &      $\mu$S/$\mu$m     \\
        Intracellular conductivity &$\sigma_1$  &  $4.55 \cdot 10^{-1}$         & $\mu$S/$\mu$m \\
        Specific membrane capacitance & $c_{m,1}$ & $9.50 \cdot 10^{-3} $ & pF$/$($\mu$m$)^2$ ($=$F/$\mbox{m}^2$)  \\
        Specific membrane resistance & $r_{m,1}$ & $1.00 \cdot 10^{5} $  & M$\Omega$($\mu$m$)^2$ \\
        Cell Radius & $\radio_1$   & 7.00          & $\mu$m \\
        Length time step & $\tau $ & $2.50\cdot 10^{-2}$ & $\mu$s \\
        Final time & $T$ & 2.50 & $\mu$s \\
        Maximum degree of spherical harmonics & $L$ & 25 &  \\
        \hline
        \end{tabular}
    \end{center}
\end{table}
\begin{figure}[tp]
    \centering
    \subfloat[$\phi_{time-ext}(t)=e^{-t}$.]{
    \includegraphics[width=0.49\textwidth]{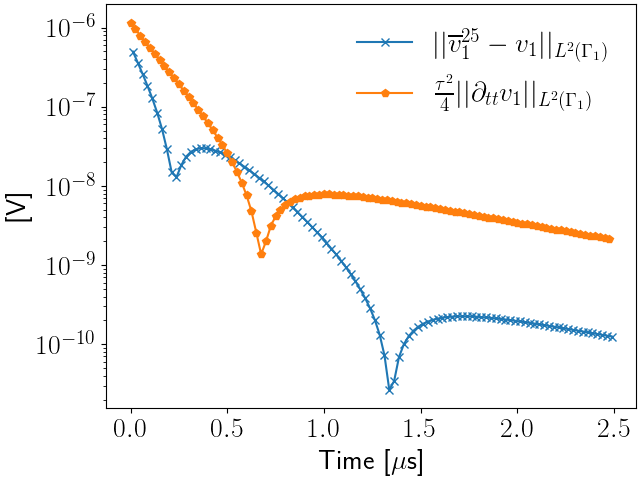}}
    \subfloat[$\phi_{time-cte}(t)=1$.]{
    \includegraphics[width=0.49\textwidth]{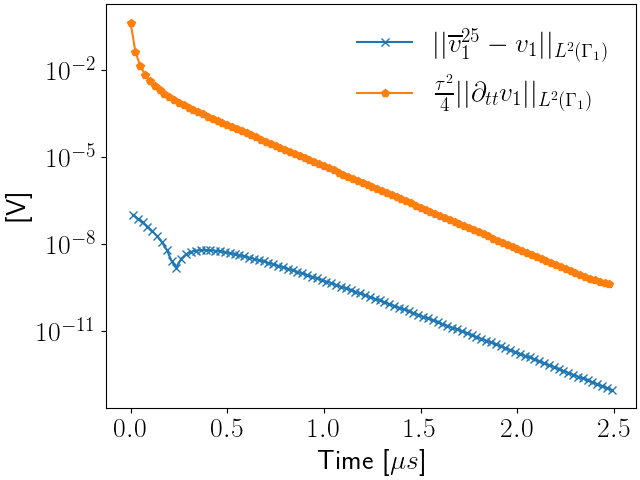}}
	\caption{Absolute error in $L^2(\Gamma_1)$ between $\overline{v}_1^{25}$ (discrete approximation) and $v_1$ (analytic solution), {as well as $\frac{\tau^2}{4} \NORM{\partial^2_{t}v_1(t)}_{L^2(\Gamma_j)}$}, plotted to verify the bound given by Theorem \ref{teo-time-converge} for the time scheme from Section \ref{sec:semiimpveri} where linear dynamics are assumed. The time step $\tau$ is $2.5\cdot10^{-2}$ $\mu$s and the rest of the parameters used are in Table \ref{time-valitadion-parameters}.}	
	\label{time-verification-results-01}
\end{figure}
\begin{figure}[tp]
	\begin{center}
		\includegraphics[width=0.49\textwidth]{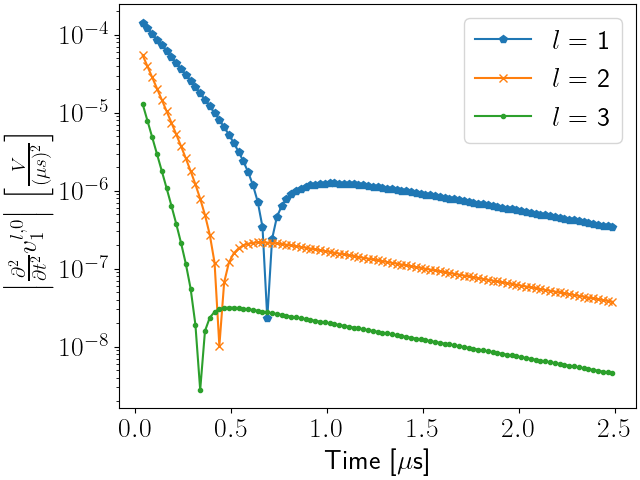}
		\caption{Absolute values of the analytically obtained second derivatives for the three most significant coefficients for the linear dynamics example from Section \ref{sec:semiimpveri} with $\phi_{time-ext}(t)=e^{-t}$. It can be seen that the dip for $l=1$ matches the dip from Figure \ref{time-verification-results-01}(a).}
		\label{plotderivative}
	\end{center}
\end{figure}
\begin{figure}[tp]
	\centering
    \subfloat[$\phi_{time-exp}(t)=e^{-t}$.]{
    \includegraphics[width=0.45\textwidth]{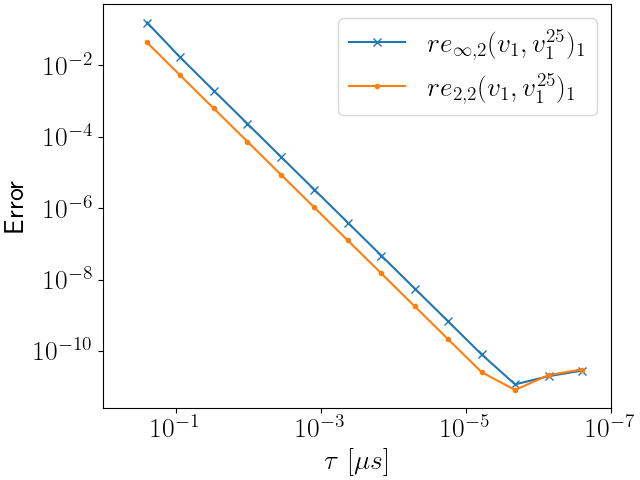}}
    \subfloat[$\phi_{time-cte}(t)=1$.]{
    \includegraphics[width=0.45\textwidth]{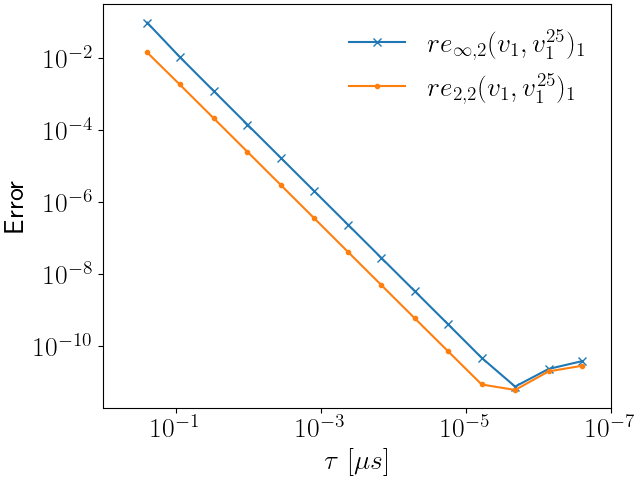}}
	\caption{Error convergence for diminishing time steps $\tau$ for the time scheme in Section \ref{sec:semiimpveri} where linear dynamics are assumed. Slopes on the log-log plot show error converges as $\tau^2$. Relative errors $re_{\infty,2}(v_1,v_1^{25})_1$ and $re_{2,2}(v_1,v_1^{25})_1$ are given in \eqref{REinf2} and \eqref{re_{2}2}, respectively. Simulation parameters can be found in Table \ref{time-valitadion-parameters}.}
	\label{time-stepconvergence}
\end{figure}

Finally, Figure \ref{time-stepconvergence} presents the relative error in time and space for decreasing values of $\tau$. We compute the error using two norms: an approximation of the $C^0 \left( (0,T), L^2(\Gamma_1) \right)$-norm taking the maximum value at each mid-step computed \eqref{REinf2}, and an approximation of the $L^2 \left( (0,T), L^2(\Gamma_1) \right)$-norm, using a composite trapezoidal rule with the computed mid-steps \eqref{re_{2}2}. We observe that the slope of the errors in the log-log plot is close to two, therefore the error decreases as $\tau^2$.

\subsection{Numerical Results for a Single Cell with Nonlinear Dynamics}
After having verified our numerical scheme for the linear dynamics, we now study the nonlinear dynamics for a single cell (Problem \ref{dynamic-problemEPBIE}). Note that in \cite[Theorem 6.14]{HJH18} error estimates are given in 2D for the Hodgkin-Huxley model. The estimates depend on four terms. The first two are the norms of the difference between  initial conditions and approximated ones used in the computations. The third error term is related to the spatial discretization, where a spectral basis in 2D is used, and this term is proved to decay exponentially with the total number of functions in the spatial discretization basis. Finally, the fourth error term is due to the time approximation, converging as $\tau^2$. Here, we expect a similar behavior. In other words, fixing the maximum degree of spherical harmonics $L$ used in the space discretization and decreasing the length of the time step $\tau$, we expect to see the error converging to a constant depending on $L$. Similarly, if we fix $\tau$ and increase $L$, we expect the error to converge to a constant depending on $\tau$.

\subsubsection{Time convergence for a fixed $L$} \label{sec:timeconv}

We use the parameters presented in Table \ref{non-linear-parameters-1sphere} to solve the non-linear discrete Problem \ref{dynamic-problem-EP-disc}, with external applied potential $\phi_e = 5 z \cdot 10^{-2}$, and initial conditions equal to zero. Since we no longer possess an analytic solution for comparison, we check for convergence as time steps become smaller. We remark that $L$ is fixed, and we use $L=1$, along with $L_c=2$.

Table \ref{non-linear-tau-convergence-1sphere} displays the error norms between two successively refined solutions for different time steps. These results show a convergence rate of one as the time step decreases, and thus we do not obtain the same as in \cite{HJH18}. This is due to the lesser regularity in time of the solutions, $Z_1$ is not twice differentiable as can be noticed from \eqref{ODE}. In Figure \ref{non-linear-evolution-1sphere}, we plot the evolution of the transmembrane potential $v_1^1$ for three different values of $\tau$. Though the solution shapes are similar, peaks appear at different locations and coincide as the time step decreases. Specifically, between $\tau=2.6\cdot 10^{-3} \mu$s and $\tau=2.6\cdot 10^{-4}\mu$s, there is a delay of less than $1.6\cdot 10^{-1} \ \mu$s, while between $\tau=2.6\cdot 10^{-4}\mu$s and $\tau=2.6\cdot 10^{-5}\mu$s the delay is less than $1.7\cdot10^{-2} \ \mu$s.
\begin{table}[tp]
\footnotesize
        \caption{Parameters used for the simulation of a single cell with non-linear dynamics \eqref{nonlinear-condition} in Section \ref{sec:timeconv} when studying the time convergence with fixed $L$. Parameters used are {found in} \protect\cite[Table 1]{KavianLeguebeea2014}.}
        \label{non-linear-parameters-1sphere}
    \begin{center}
        \begin{tabular}{|l|l|l|l|}
        \hline
        Parameter & Symbol      & Values         & Unit \\
        \hline
        Cell Radius & $\radio_1$   & $1.00\cdot 10^1$           & $\mu$m \\
        Time part of $\phi_e$ & $\phi_{time}$             &    $1.00$    &    \\
        Spatial part of $\phi_e$ & $\phi_{spatial}$ & 5 $z \cdot 10^{-2}$  &    V     \\
        Extracellular conductivity &$\sigma_0$   & 5.00  &      $\mu$S/$\mu$m     \\
        Intracellular conductivity &$\sigma_1$  &  $4.55 \cdot 10^{-1}$         & $\mu$S/$\mu$m \\
        Lipid surface conductivity &$S_{L,1}$  &  $1.90\cdot 10^{-6}$         & $\mu$S/($\mu$m$)^2$ \\
        Irreversible surface conductivity &$S_{ir,1}$  &  $2.50\cdot 10^{2}$         & $\mu$S/($\mu$m$)^2$ \\
        Specific membrane capacitance & $c_{m,1}$ & $9.50 \cdot 10^{-3} $ & pF$/$($\mu$m$)^2$ \\
        Transmembrane potential threshold & $V_{rev,1}$ & 1.50 & V  \\
        Electropermeabilization switch speed & $k_{ep,1}$ & $4.00\cdot 10^1$ & $\mbox{V}^{-1}$  \\
        Characteristic time of electropermeabilization & $\tau_{ep,1}$ & 1.00 & $\mu$s  \\
        Characteristic resealing time & $\tau_{res,1}$ & $1.00\cdot 10^3$ & $\mu$s  \\
        Final time & $T$ & $2.60\cdot10^1$ & $\mu$s \\
        Maximum degree of spherical harmonics & $L$ & 1 &  \\
        Quadrature degree & $L_c$ & 2 &  \\
        \hline
        \end{tabular}
    \end{center}
\end{table}
\begin{table}[tp]
\footnotesize
\caption{Error convergence for the nonlinear problem with one cell from Section \ref{sec:timeconv} for fixed $L$. Computed norms are the difference between two successive solutions. Parameters used are in Table \ref{non-linear-parameters-1sphere}.}
        \label{non-linear-tau-convergence-1sphere}
    \begin{center}
        \begin{tabular}{|l|l|c|c|}
        \hline
        $\tau_i$ & Value [$\mu$s]      & $\max_{t\in [0,T]}\left\|v_1^{1, \tau_{i+1}} - v_1^{1, \tau_{i}} \right\|_{L^2(\Gamma_1)}$         & $\max_{t\in [0,T]}\left\|Z_1^{1, \tau_{i+1}} - Z_1^{1, \tau_{i}} \right\|_{L^2(\Gamma_1)}$  \\
        \hline
        $\tau_1$ & $2.6 \cdot 10^{-3}$ &   - &  -  \\
        $\tau_2$ & $2.6 \cdot 10^{-4}$ &   $8.8 \cdot 10^{0}$ &  $4.64 \cdot 10^{-3}$ \\
        $\tau_3$ & $2.6 \cdot 10^{-5}$ &   $9.0 \cdot 10^{-1}$ & $3.02 \cdot 10^{-4}$  \\
        $\tau_4$ & $2.6 \cdot 10^{-6}$ &  $9.7 \cdot 10^{-2}$ & $3.59\cdot 10^{-5}$   \\
        $\tau_5$ & $2.6 \cdot 10^{-7}$ &  $9.7 \cdot 10^{-9}$ & $3.15 \cdot 10^{-6}$   \\
        \hline
        \end{tabular}
    \end{center}
\end{table}
\begin{figure}[tp]
    \centering
    {\includegraphics[width=0.45\textwidth]{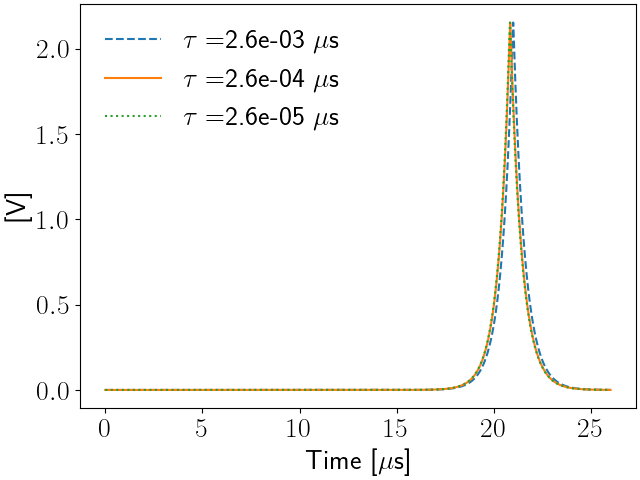}}
    {\includegraphics[width=0.45\textwidth]{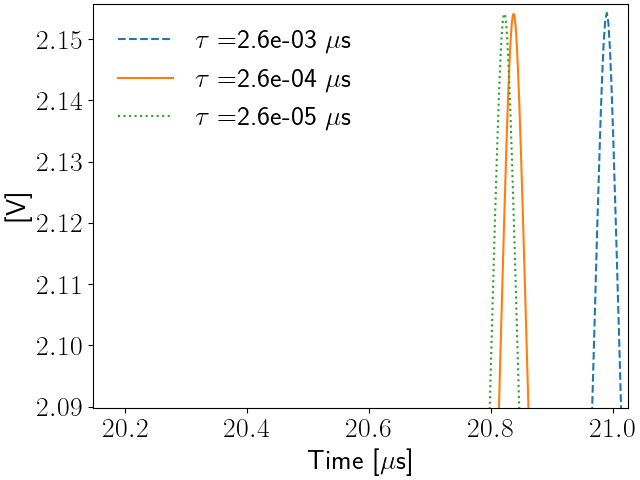}}
	\caption{Evolution of $v_1^1$ at the north pole of the cell ($\theta=0$) for different lengths of time step $\tau$ illustrating the time convergence for fixed $L$, Section \ref{sec:timeconv}. The image at the right is zoomed near to the maximum value of $v_1^1$. Parameters employed are given in Table \ref{non-linear-parameters-1sphere}.}	
	\label{non-linear-evolution-1sphere}
\end{figure}
\subsubsection{Spatial convergence with nonlinear dynamics}\label{sec:spaceconv}
We now present numerical results for different maximum degrees of the spherical harmonics, $L=51$ and $L \in [1,2,...,36]$, computed with $L_c=150$. Given that we use a spectral discretization in space, we expect an exponential decrease in the error when increasing the maximum degree\footnote{The parameters used are provided in Table \ref{non-linear-one}. Notice that extra- and intracellular conductivities have different values from the previous simulations, and were changed to obtain a response of the impulse sooner.} $L$---recall that the number of spatial discretization functions basis is $(L+1)^2$. The external applied potential is $\phi_e = 5 z \cdot 10^{-2}$ until $t=5$ and equal to zero thereafter. Initial conditions are set to zero, and the length of the time step used is $\tau \approx 2.4\cdot10^{-3}$.

\begin{table}[tp]
\footnotesize
\caption{Parameters used in the numerical simulations in Sections \ref{sec:spaceconv} and \ref{sec:multiplecells}, with the non-linar dynamics of the electropermeabilization model. The specific choice of extra- and intracellular conductivities, different from the previous simulations, allow us to obtain a response of the impulse sooner in time. The rest of the parameters are from \protect\cite[Table 1]{KavianLeguebeea2014}. The external applied potential used is equal to zero after $t=5$ $\mu$s.}
        \label{non-linear-one}
    \begin{center}
        \begin{tabular}{|l|l|l|l|}
        \hline
        Parameter & Symbol      & Values         & Unit \\
        \hline
        Cell Radius & $\radio_1$   & $1.00\cdot10^1$           & $\mu$m \\
        External applied potential & $\phi_{e}$ & 5 $z \cdot 10^{-2}$  &    V     \\
        Extracellular conductivity &$\sigma_0$   & $1.50\cdot 10^1$  &      $\mu$S/$\mu$m     \\
        Intracellular conductivity &$\sigma_1$  &  1.50         & $\mu$S/$\mu$m \\
        Specific membrane capacitance & $c_{m,1}$ & $9.50 \cdot 10^{-3} $ & pF$/$($\mu$m$)^2$ ($=$F/$\mbox{m}^2$)  \\
        Lipid surface conductivity &$S_{L,1}$  &  $1.90\cdot 10^{-6}$         & $\mu$S/($\mu$m$)^2$ \\
        Irreversible surface conductivity &$S_{ir,1}$  &  $2.50\cdot 10^{2}$         & $\mu$S/($\mu$m$)^2$ \\
        Transmembrane potential threshold & $V_{rev,1}$ & 1.50 & V  \\
        Electropermeabilization switch speed & $k_{ep,1}$ & $4.00\cdot10^1$ & $\mbox{V}^{-1}$  \\
        Characteristic time of electropermeabilization & $\tau_{ep,1}$ & 1.00 & $\mu$s  \\
        Characteristic resealing time & $\tau_{res,1}$ & $1.00\cdot10^3$ & $\mu$s  \\
        Final time & $T$ & $1.00\cdot10^1$ & $\mu$s \\
        \hline
        \end{tabular}
    \end{center}
\end{table}

We compute the relative errors between $v_1^L$ and $v_1^{51}$, and between $Z_1^L$ and $Z_1^{51}$. The results are shown in Figure \ref{convergence-one-non-linear}. The plots are in a log-linear scale, and the errors tends to form a straight line with the slope of order $10^{-2}$, which suggests an exponential rate of convergence. Recall that  in our case $\beta$ (cf.~\eqref{beta-used}) is only $C^0$-continuous due to the discontinuity of the derivative at the origin and $\Zz_1$ is only twice differentiable in time \eqref{ODE}, worsening the rate of convergence. However, as shown in \cite{stephan_convergence_1989}, the numerical method presented should converge faster than a first-order boundary element method and twice as fast with respect to the number of functions used to construct the approximation in the worst case. While the obtained $Z_1$ is an even function in space, $v_1$ is an odd one.  Thus, the nonlinear current is an odd function in space. Since the external applied potential is an odd function, we expect $v_1$ to have an odd component, while $Z_1$ is defined by an ordinary differential equation that takes $v_1$ into an even function. Finally, in Figure \ref{vresult} we plot the evolution in time of $v_1^{17}$, $v_1^{24}$, $v_1^{35}$, and $v_1^{51}$ at the north pole. The differences between the results are more noticeable after the peak of the potential and when the cell tries to stabilize it.

\begin{figure}[tp]
    \centering
    \subfloat[Relative norms for $v^L_1$.]{    \includegraphics[width=0.45\textwidth]{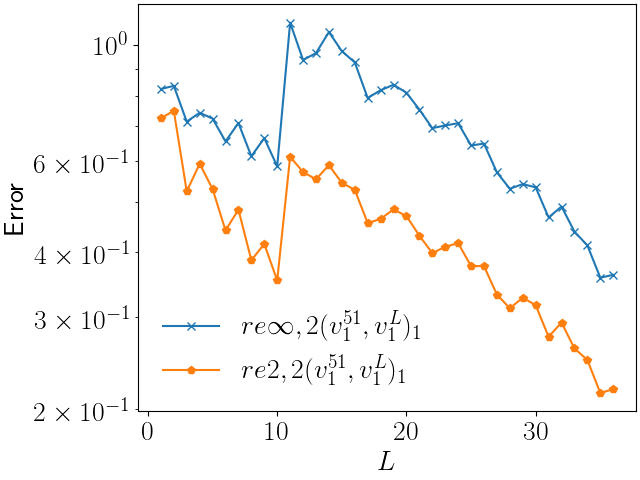}}
    \subfloat[Relative norms for $Z^L_1$.]{
    \includegraphics[width=0.44\textwidth]{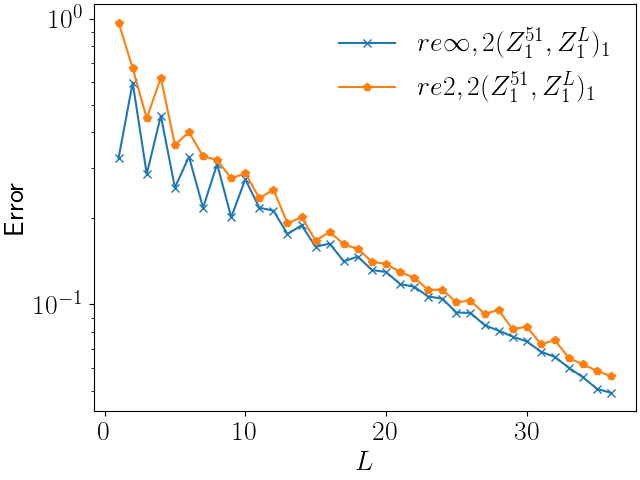}}
	\caption{Spatial convergence for the nonlinear dynamics of Section \ref{sec:spaceconv}. Relative norms are computed against an overkill of $L=51$. On the left, results for $v_1^L$, while on the right  $Z_1^L$ is displayed with  time step $\tau\approx 2.4\cdot10^{-2} \ \mu$s. The $x-$axis indicates the maximum degree used for discretization. Convergence starts from $L=11$. Plots are in log-linear scale and error tends to form a straight line with slope of approximately $-10^{-2}$, i.e.~exponential convergence. Parameters are given in Table \ref{non-linear-one}. The expressions for the relative errors $re_{\infty,2}(v_1^{51},v_1^{L})_1$ and $re_{2,2}(v_1^{51},v_1^{L})_1$ are given in \eqref{REinf2} and \eqref{re_{2}2}, respectively.}
    \label{convergence-one-non-linear}
\end{figure}

\begin{figure}[tp]
    \centering
    \subfloat[]{    \includegraphics[width=0.43\textwidth]{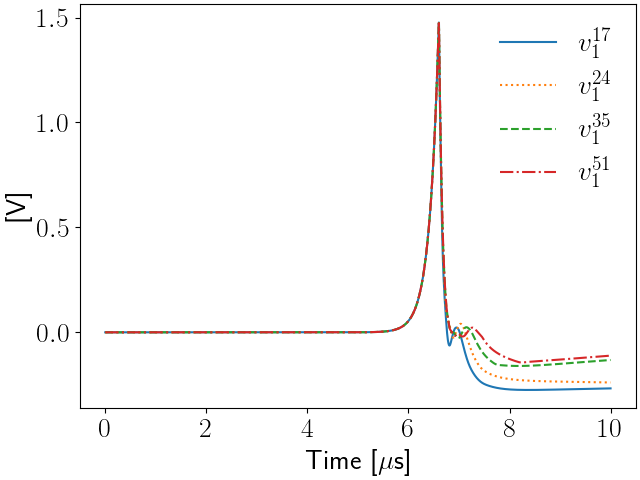}}
    \subfloat[Zoom at peak.]{
    \includegraphics[width=0.43\textwidth]{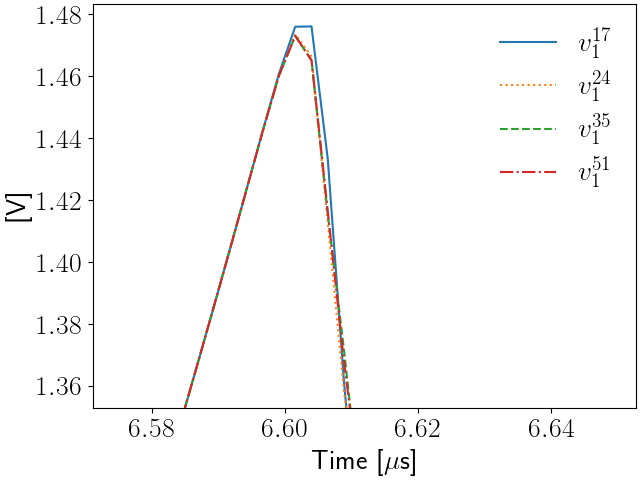}}
	\caption{Evolution of the transmembrane potentials $v_1^{17}$, $v_1^{24}$, $v_1^{35}$ and $v_1^{51}$ at the north pole of the cell ($\theta=0$) obtained in Section \ref{sec:spaceconv} where the spatial convergence for one cell in the nonlinear case is studied. The time step used is $\tau\approx 2.4\cdot10^{-2} \ \mu$s, with parameters given in Table \ref{non-linear-one}.}
 \label{vresult}
\end{figure}

\subsection{Results with multiple cells}\label{sec:multiplecells}

In previous sections, the convergence of the numerical method was studied for a single cell. We proceed now with the case of multiple cells to perform five experiments in the nonlinear case. The examples presented highlight how the distance among  cells affects the results as all cell conductivities are set to the same value $\sigma_1$.
\begin{itemize}
    \item {\bf Example 3}: Three cells aligned along the $x$-axis and far from each other, with distance between cells $18\radio_1$. 

    \item {\bf Example 4}: Three cells aligned along the $x$-axis,  near from each other, with distance between cells $\frac{1}{2}\radio_1$.

    \item {\bf Example 5}: Eight cells aligned in a cubic lattice, the nearest distance between two cells is $\frac{1}{2}\radio_1$, the first sphere is at the origin.
\end{itemize}
\begin{table}[tp]
\footnotesize
\caption{Center positions for Examples 3 and 4 from Section \ref{sec:multiplecells}, where nonlinear dynamics are simulated.}
\begin{center}
        \begin{tabular}{|c|c|c|c|c|c|c|} \hline
        Center position   & Symbol   & Example 3 & Example 4  & Unit \\ \hline
        Cell 1  & $\mathbf{p_1}$  &  (0, 0, 0) & (0, 0, 0) &   $\mu$m     \\
        Cell 2 & $\mathbf{p_2}$  &  (200, 0, 0) &(25, 0, 0) & $\mu$m     \\
        Cell 3  & $\mathbf{p_3}$  &  (-200, 0, 0)  &(-25, 0, 0) & $\mu$m     \\ \hline
        \end{tabular}
        \label{non-linear-3-position-parameters}
    \end{center}    
\end{table}
\begin{table}[tp]
\footnotesize
\caption{Positions of cells in Example 5 from Section \ref{sec:multiplecells}, where nonlinear dynamics are simulated.}
\begin{center}
        \label{non-linear-8-position-parameters}
        \begin{tabular}{|c|c|c|c|c|c|} \hline
        Center position   & Symbol   & Value in $\mu$m & Center position   & Symbol   & Value in $\mu$m  \\ \hline
        Cell 1  & $\mathbf{p_1}$  &  (0, 0, 0) & Cell 5  & $\mathbf{p_5}$  &  (0, 0, 25)   \\
        Cell 2 & $\mathbf{p_2}$  &  (25, 0, 0) &Cell 5  & $\mathbf{p_6}$  &  (25, 0, 25)  \\
        Cell 3  & $\mathbf{p_3}$  &  (0, 25, 0)  &Cell 7  & $\mathbf{p_7}$  &  (0, 25, 25)    \\
        Cell 4  & $\mathbf{p_4}$  &  (25, 25, 0)  &Cell 8  & $\mathbf{p_8}$  & (25, 25, 25)  \\ \hline
        \end{tabular}
    \end{center}    
\end{table}
\begin{figure}[tp]
    \centering
    {\includegraphics[width=0.40\textwidth]{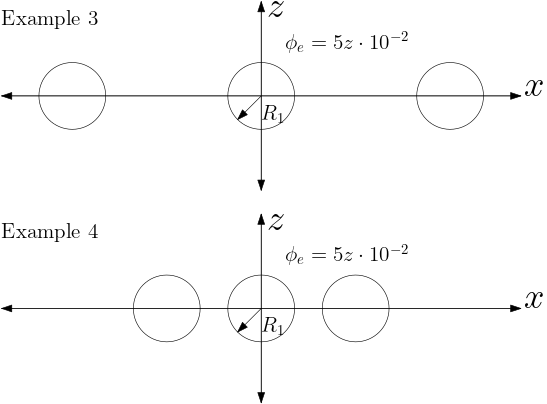}}
	\caption{Illustration of cells positions for Examples 3 and 4 in Section \ref{sec:multiplecells}.}
	\label{diagram}
\end{figure}

Cell radii and physical parameters used for Examples 3--5 are presented in Table \ref{non-linear-one}. Extra- and intracellular conductivity values were changed so as to obtain a response sooner. Cells centers in Examples 3 and 4 are given in Table \ref{non-linear-3-position-parameters} and sketched in Figure \ref{diagram}, while those in Example 5 are located at the corners of a cube of length 25 $\mu$m. (cf.~Table \ref{non-linear-8-position-parameters}). Throughout initial conditions are set to zero. The external applied potential in Examples 3--5 is $\phi_e = 5z \cdot 10^{-2}$ until $t=5$ $\mu$s and zero thereafter.

In what follows, we present  results for a time step $\tau\approx 6.1 \cdot 10^{-4}$. The maximum degree of spherical harmonics used for Examples 3 and 4 is $L=35$, while for Example 5 we set $L=25$. Quadrature degree used in all examples is $L_c=100$. Figures \ref{non-linear-evolution-example5}, \ref{non-linear-evolution-example6} and \ref{non-linear-evolution-example9} showcase the evolution of the transmembrane potentials $v_j^{L}$ and the variables $Z_j^{L}$ for each cell at their north pole.

\begin{figure}[tp]
    \centering
    {\includegraphics[width=0.43\textwidth]{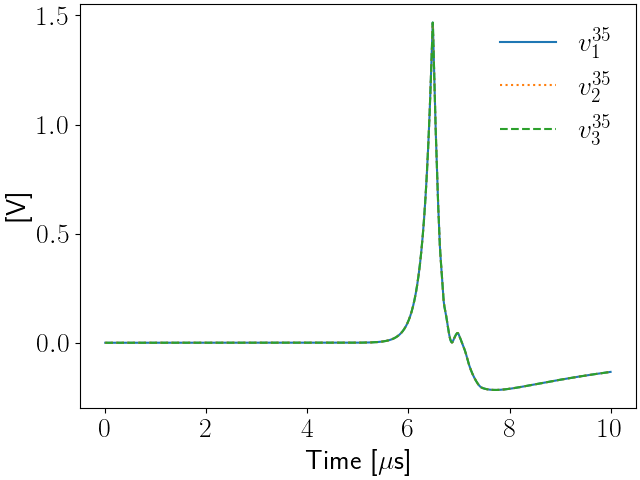}}
    {\includegraphics[width=0.43\textwidth]{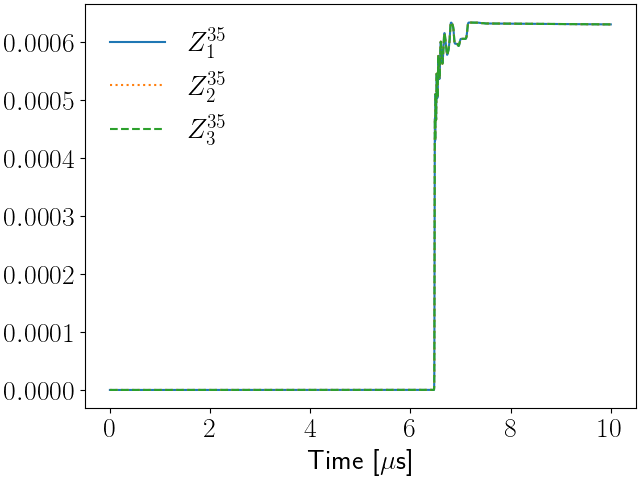}}
	\caption{Evolution of $v_j^{35}$ and $Z_j^{35}$ at the north pole of each $j$ cell ($\theta=0$), from Example 3 in Section \ref{sec:multiplecells}. The time step is $\tau\approx 6.1 \cdot 10^{-4}$. Parameters employed are found in Tables \ref{non-linear-one} and \ref{non-linear-3-position-parameters}.}
	\label{non-linear-evolution-example5}
\end{figure}
\begin{figure}[tp]
    \centering
    {\includegraphics[width=0.43\textwidth]{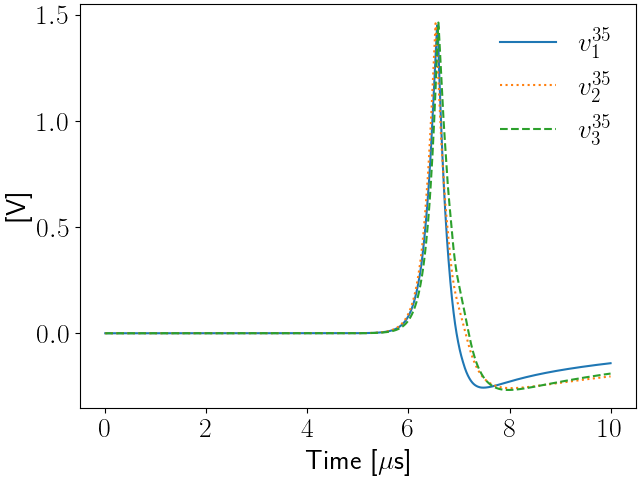}}
    {\includegraphics[width=0.43\textwidth]{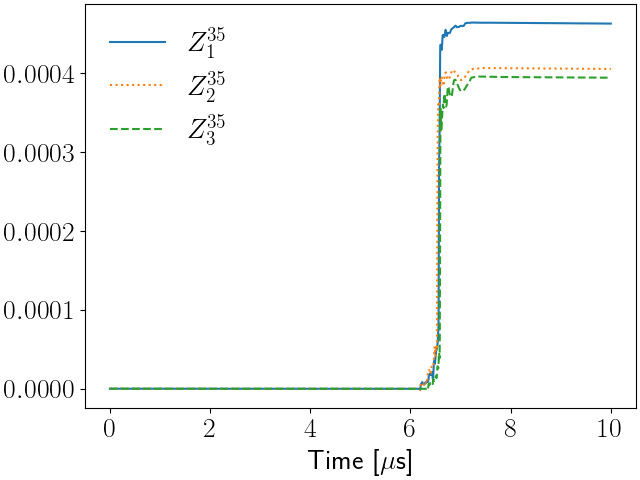}}
	\caption{Evolution of $v_j^{35}$ and $Z_j^{35}$ at the north pole of each $j$ cell ($\theta=0$), from Example 4 from Section \ref{sec:multiplecells}. Cells are near each other and the interaction among them influences the transmbembrane potential $v_j^{35}$ and $Z_j^{35}$ (cf.~Example 3 in contrast). The only difference between Examples 3 and 4 is the distance between successive cells. The time step is $\tau\approx 6.1 \cdot 10^{-4}$, and the parameters employed are given in Tables \ref{non-linear-one} and \ref{non-linear-3-position-parameters}.}
	\label{non-linear-evolution-example6}
\end{figure}
\begin{figure}[tp]
    \centering
    {\includegraphics[width=0.43\textwidth]{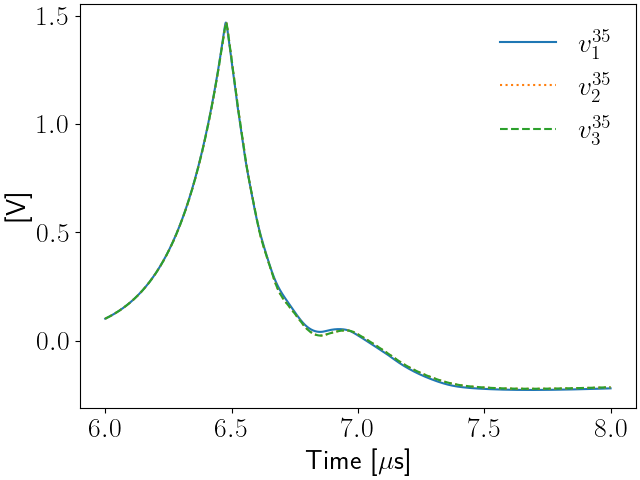}}
    {\includegraphics[width=0.43\textwidth]{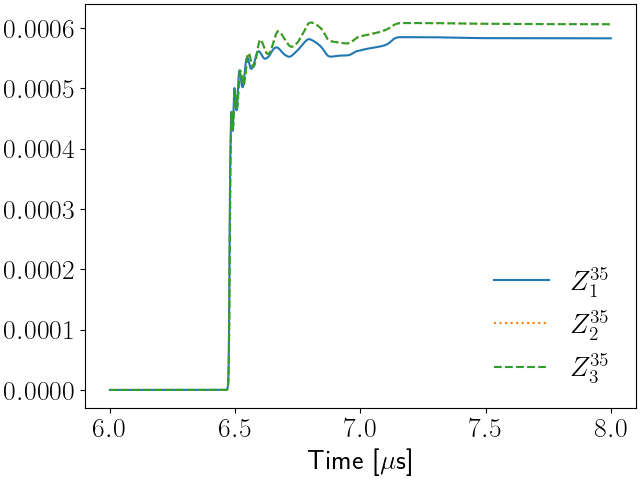}}
	\caption{Evolution of $v_j^{35}$ and $Z_j^{35}$ at the north pole of each $j$ cell ($\theta=0$), from a simulation with a refinement in time of Example 4 from Section \ref{sec:multiplecells}. Cells are near each other and the interaction among them influences the transmbembrane potential $v_j^{35}$ and $Z_j^{35}$ (cf.~Example 3 in contrast). The only difference between Examples 3 and 4 is the distance between successive cells. The time step is $\tau\approx 4.9 \cdot 10^{-4}$.}
 \label{non-linear-evolution-example6-2}
\end{figure}
\begin{figure}[tp]
    \centering
    {\includegraphics[width=0.42\textwidth]{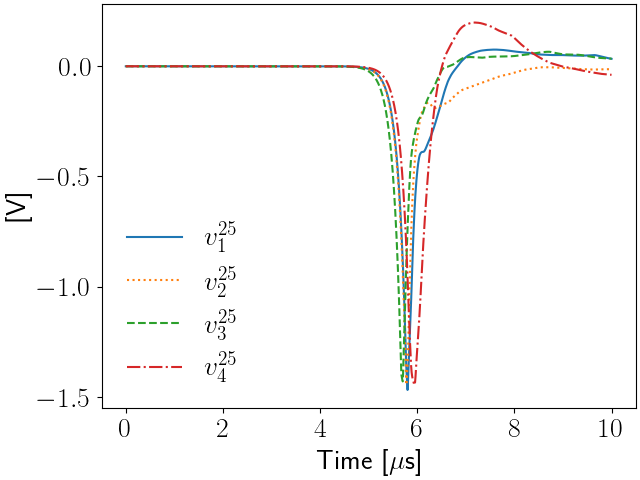}}
    {\includegraphics[width=0.42\textwidth]{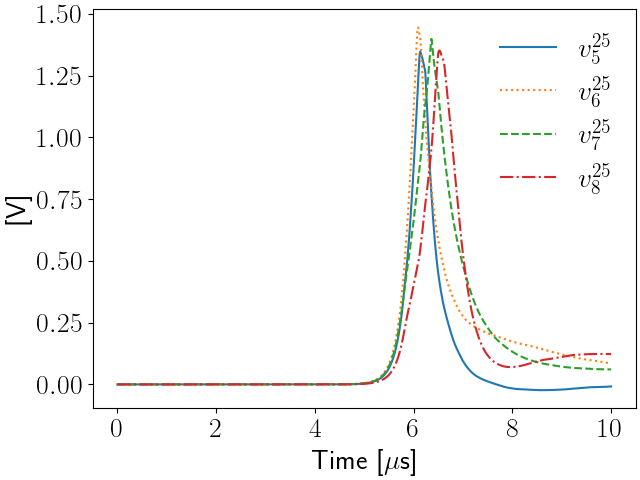}}\\
    {\includegraphics[width=0.40\textwidth]{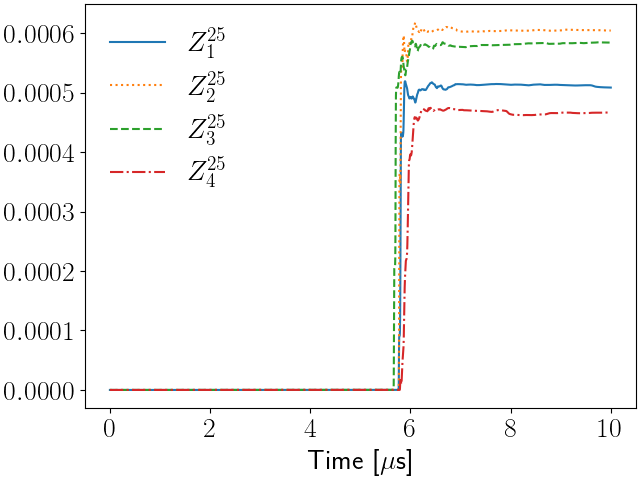}}
    {    \includegraphics[width=0.40\textwidth]{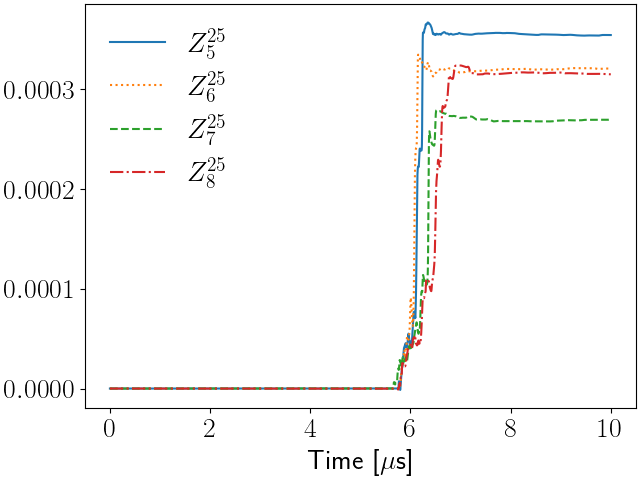}}
	\caption{Evolution of $v_j^{25}$ and $Z_j^{25}$ at the north pole of each cell ($\theta=0$), from Example 5 from Section \ref{sec:multiplecells}. The first four cells are in the plane $z=0$, while the others are in the plane $z=25$. The time step is $\tau\approx 6.1 \cdot 10^{-4} \mu$s. The rest of parameters employed are given in Tables \ref{non-linear-one} and \ref{non-linear-3-position-parameters}.}	
	\label{non-linear-evolution-example9}
\end{figure}
\begin{figure}[tp]
    \centering
    \subfloat[$t=5.75\ \mu$s.]{
    \includegraphics[width=0.40\textwidth]{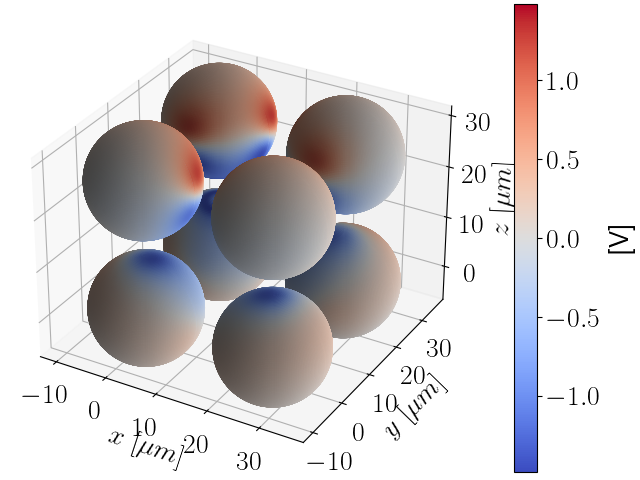}}
    \subfloat[$t=6\ \mu$s.]{
    \includegraphics[width=0.40\textwidth]{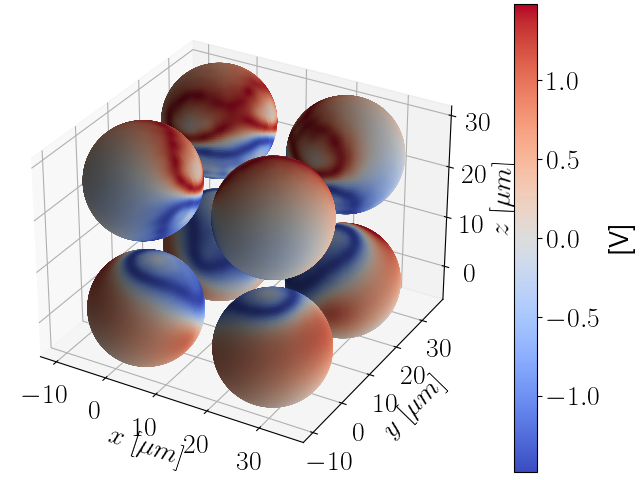}}\\
    \subfloat[$t=6.5\ \mu$s.]{
    \includegraphics[width=0.40\textwidth]{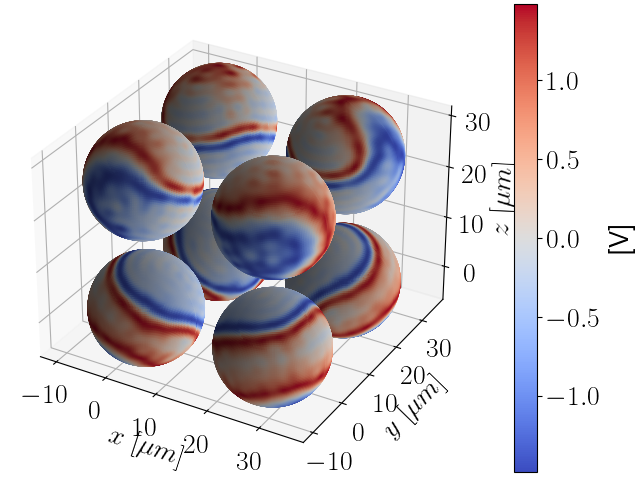}}
    \subfloat[$t=7\ \mu$s.]{
    \includegraphics[width=0.40\textwidth]{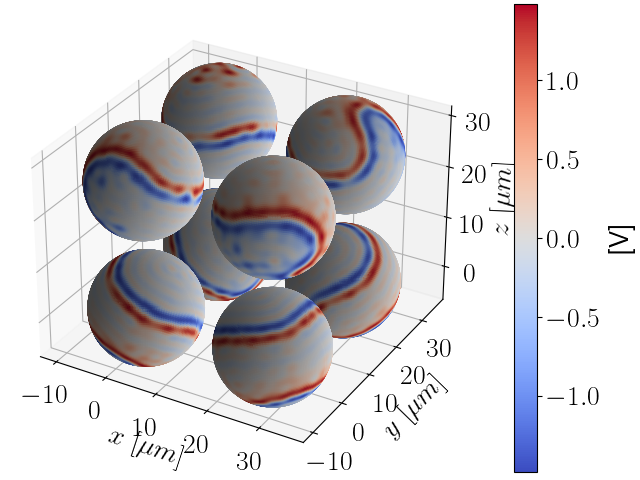}}
	\caption{Transmembrane voltages $v_j^{25}$ obtained in Example 9 of Section \ref{sec:multiplecells} at different times. The length of the time step is $\tau\approx 6.1 \cdot 10^{-4}$. Parameters employed are given in Tables \ref{non-linear-one} and \ref{non-linear-3-position-parameters}.}
 \label{surfaceplot}
\end{figure}

In Example 3, the cells are aligned along the $x$-axis and $\phi_e=5z \cdot 10^{-2}$. Therefore, the external excitation from $\phi_e$ is the same for the three spheres, i.e. the contribution of $\phi_e$ to the right-hand side is the same for each sphere. Since the cells  are relatively far from each other, there is almost no interaction among them and the potentials $v_j^{35}$ and $Z_j^{35}$ look similar, for all $j$ (see Figure \ref{vresult}). In Example 4, we take the same parameters as in Example 3, but the distance between cells is reduced to $\frac{1}{2}R_1$. Hence, the interaction among cells is stronger, and, as expected, the shapes of the potentials $v_j^{35}$ and $Z_j^{35}$ change (see Figure \ref{non-linear-evolution-example6}). One should compare these results with the previous example in Figure \ref{non-linear-evolution-example5}. Due to the symmetry and the form of $\phi_e=5z\cdot 10^{-2}$, cells 2 and 3 should have the same response at the north pole. However, they are slightly different, hinting at further time step refinement. To check this, in Figure \ref{non-linear-evolution-example6-2} we present results for a second simulation with a refined time step, between 6 $\mu$s and 8 $\mu$s, using the previous simulation at that time as initial condition. One can immediately see that the transmembrane potentials recover the stated symmetry. 

Finally, in Example 5 eight cells close to each other are simulated. In Figure \ref{non-linear-evolution-example9}, the corresponding transmembrane voltage $v_j^{25}$ and $Z_j^{25}$ at the north pole are presented. The cells with the centers in the plane $z=0$ show similar response---see Table \ref{non-linear-8-position-parameters} for the center position of each cell---, while the cells with  centers in the plane $z=25$ have similar response too while differing from cells beneath them. Figure \ref{surfaceplot} shows six snapshots of the transmembrane voltages for the eight cells. The transmembrane voltage starts changing earlier on parts of the surface closest to the rest of the cells.

\section{Conclusions and future work}
We studied the electropermeabilization of disjoint cells following the nonlinear dynamics from \cite{KavianLeguebeea2014} and recast the volume boundary value problem via a MTF to obtain a parabolic system of boundary integral equations on the cell membrane. This extends signficantly the numerical method presented in \cite{HJH18}. Still, the multistep semi-implicit time scheme though stable requires relatively small time steps with low convergence rates due to the poor regularity in time.
For simplicity, we assumed spherical cells but other shapes can be easily considered with similar spectral convergence rates in space. 

One can easily change the dynamics' model as long as it only involves the nonlinear term and/or variables that exclude the transmembrane potential. In this case, only the right-hand side of the system to be solved \eqref{discrete-system} and the equations corresponding to the additional variables.  Further improvements to the numerical method to be implemented in the future are matrix compression and parallelization techniques, along with an efficient solver for linear systems at each time step.

\FloatBarrier

\bibliographystyle{acm} 
\bibliography{biblio}

\end{document}